\documentclass[
a4paper,reqno]{amsart}
\usepackage[T1]{fontenc}
\usepackage[english]{babel}
\usepackage{amssymb,bbm,enumerate}
\usepackage{xcolor}
\usepackage[linktocpage=true,colorlinks=true, linkcolor=blue, citecolor=red, urlcolor=green]{hyperref}

\usepackage{float}
\restylefloat{table}

\renewcommand{\phi}{\varphi}
\renewcommand{\ker}{\Ker}
\renewcommand{\Re}{\operatorname{Re}}

\newcommand{\bb}[1]{\mathbb{#1}}
\newcommand{\mc}[1]{\mathcal{#1}}
\newcommand{\mf}[1]{\mathfrak{#1}}
\newcommand{\mb}[1]{\mathbb{#1}}
\newcommand{\mbbm}[1]{\mathbbm{#1}}

\newcommand{\beq}{\begin{equation}}
\newcommand{\eeq}{\end{equation}}
\newcommand{\e}{\varepsilon}
\DeclareMathOperator{\at}{at}

\DeclareMathOperator{\Mat}{Mat}

\DeclareMathOperator{\diag}{diag}
\DeclareMathOperator{\tr}{Tr}
\DeclareMathOperator{\res}{Res}
\DeclareMathOperator{\ad}{ad}

\DeclareMathOperator{\Ker}{Ker}

\DeclareMathOperator{\rank}{rank}

\theoremstyle{plain}
\newtheorem{theorem}{Theorem}[section]
\newtheorem{lemma}[theorem]{Lemma}

\newtheorem{proposition}[theorem]{Proposition}
\newtheorem{corollary}[theorem]{Corollary}

\theoremstyle{definition}
\newtheorem{definition}[theorem]{Definition}
\newtheorem{example}[theorem]{Example}

\theoremstyle{remark}
\newtheorem{remark}[theorem]{Remark}

\numberwithin{equation}{section}

\definecolor{light}{gray}{.9}

\usepackage{float}
\restylefloat{table}
\usepackage{mathtools}
\usepackage{tikz}
\usetikzlibrary{chains}

\tikzset{node distance=2em, ch/.style={circle,draw,on chain,inner sep=2pt},chj/.style={ch,join},every path/.style={shorten >=4pt,shorten <=4pt},line width=1pt,baseline=-1ex}

\let\dlabel=\alabel

\newcommand{\dnode}[2][chj]{%
\node[#1,label={below:\dlabel{#2}}] {};
}

\newcommand{\dnodebr}[1]{%
\node[chj,label={below right:\dlabel{#1}}] {};
}

\newcommand{\dydots}{%
\node[chj,draw=none,inner sep=1pt] {\dots};
}

\title{Bethe Ansatz and the Spectral Theory of
affine Lie algebra--valued connections. \\ The simply--laced case}

\author{Davide Masoero, Andrea Raimondo, Daniele Valeri}

\address{Grupo de F\'isica Matem\'atica da Universidade de Lisboa,
Av. Prof. Gama Pinto 2, 1649-003 Lisboa, Portugal.}
\email{dmasoero@gmail.com}

\address{Dipartimento di Matematica e Applicazioni, Universit\`a degli Studi di Milano-Bicocca,
Via Cozzi 53, 20125 Milano, Italy.}
\email{andrea.raimondo@unimib.it}

\address{Yau Mathematical Sciences Center, Tsinghua University, 100084 Beijing, China.}
\email{daniele@math.tsinghua.edu.cn}

\begin{document}

\pagestyle{plain}

\begin{abstract}
We study the ODE/IM correspondence for ODE associated to
$\widehat{\mf g}$-valued connections, for a simply-laced Lie algebra $\mf g$.
We prove that subdominant solutions to the ODE defined
in different fundamental
representations satisfy a set of quadratic equations
called $\Psi$-system.
This allows us to show that the
generalized spectral determinants satisfy the 
Bethe Ansatz equations.
\end{abstract}

\maketitle

\tableofcontents

\section{Introduction}\label{sec:intro}
The ODE/IM correspondence, developed since the seminal papers \cite{doreytateo98,bazhanov01}, concerns the
relations between the generalized spectral problems of linear Ordinary Differential Equations (ODE) and two dimensional
Integrable Quantum Field Theories (QFT), or continuous limit of Quantum Integrable Systems. Accordingly,
the acronym IM usually refers to Integrable Models (or Integrals of Motion). More recently,
it was observed that the correspondence can also be interpreted as a  relation between Classical and Quantum
Integrable Systems \cite{lukyanov10}, or as an instance of the Langlands duality \cite{FF11}.

The original ODE/IM correspondence states that the spectral determinant of certain Schr\"odinger operators coincides with 
the vacuum eigenvalue $Q(E)$ of the Q-operators of  the quantum KdV equation \cite{bazhanov97}, when $E$ is interpreted
as a parameter of the theory. This correspondence has been proved \cite{doreytateo98,bazhanov01} by showing that the
spectral determinant and the eigenvalues of the $Q$ operator solve the same functional equation, namely the Bethe Ansatz,
and they have the same analytic properties.

The correspondence has soon been generalized to higher eigenvalues of the Quantum KdV equation -- obtained modifying
the potential of the Schr\"odinger operator \cite{BLZ04} -- as well as  to Conformal Field Theory with extended ${\mc{W}}_n$-symmetries,
by considering ODE naturally associated with the simple Lie algebra $A_n$ (scalar linear ODEs of order $n+1$) \cite{dorey00,junji00}.
Within this approach, the original construction corresponds to the Lie algebra $A_1$, and the associated ODE is the Schr\"odinger equation.

After these results, it became clear that a more general picture including all simple Lie algebras --  as well as massive
(i.e. non conformal) Integrable QFT -- was missing. While the latter issue was addressed in \cite{gaiotto09,lukyanov10}
(and more recently in \cite{dorey13, dunning14, negro14}), the ODE/IM correspondence for simple Lie algebras different from $A_n$
has not yet been fully understood.

The first work proposing an extension of the ODE/IM correspondence to classical Lie algebras other than $A_n$ is \cite{dorey07},
where the authors introduce a $\mf{g}$-related scalar ODE 
as well as the important concept of $\Psi$-system. The latter is a set of quadratic relations that solutions of the
$\mf g$-related ODE are expected to satisfy in different representations.
It was then recognized in \cite{FF11} that the $\mf g$-related scalar differential
operators studied in \cite{BLZ04,dorey07} can be regarded as \textit{affine opers} associated to the Langlands dual
of the affine Lie algebra $\widehat{\mf g}$. This discovery  gave the theory a solid algebraic footing --
based on Kac-Moody algebras and Drinfeld-Sokolov Lax operators -- which was further investigated in \cite{Sun12} and that we hereby follow.

\medskip

In the present paper we establish the ODE/IM correspondence for a simply laced Lie algebra $\mf g$, in which case the affine Lie algebra
$\widehat{\mf g}$ coincides with its Langlands dual \cite{FF11}.
More concretely, we  prove that for any simple Lie algebra $\mf g$ of $ADE$ type there exists a family of spectral problems,
which are encoded by entire functions $Q^{(i)}(E)$, $i=1,\dots,n=\rank\mf g$, satisfying the $ADE$-Bethe Ansatz equations
\cite{reshetikhin87,zamo91,bazhanov02integrable}: 
\begin{equation}\label{eq:TBAintro}
\prod_{j = 1}^{n} \Omega^{\beta_jC_{ij}} \frac{Q^{(j)}\Big(\Omega^{\frac{C_{ij}}{2}}E^*\Big)}{Q^{(j)}
\Big(\Omega^{-\frac{C_{ij}}{2}}E^*\Big)}=-1
\,,
\end{equation}
for every $E^*\in\mb C$ such that $Q^{(i)}(E^*)=0$. In equation \eqref{eq:TBAintro}, $C=(C_{ij})_{i,j=1}^n$ is the Cartan matrix of the
algebra $\mf g$, while the phase $\Omega$ and the numbers $\beta_j$ are the free-parameters of the equation.
More precisely, following \cite{FF11} we introduce the $\widehat{\mf g}$-valued connection 
\begin{equation}\label{eq:Lintro}
\mc L(x,E)=\partial_x+\frac{\ell}{x}+e+p(x,E) e_0
\,,
\end{equation}
where $\ell\in\mf h$ is a generic element of the Cartan subalgebra $\mf h$ of $\mf g$, $e_0,e_1,\dots,e_{n}$ are
the positive Chevalley generators of the affine Kac-Moody algebra $\widehat{\mf g}$ and $e=\sum_{i=1}^{n}e_i$.
In addition, the potential\footnote{We stick to a potential of this form for simplicity. However all proofs work
with minor modifications for a more general potential discussed in \cite{BLZ04,FF11}.} is $p(x,E)=x^{Mh^\vee}-E$,
where $M>0$ and $h^\vee$ is the dual Coxeter number of $\mf g$.

After \cite{Sun12}, for every fundamental representation \footnote{Actually $V^{(i)}$ is an evaluation representation of the $i$-th
fundamental representation of $\mf g$.} $V^{(i)}$ of $\widehat{\mf{g}}$ we consider the differential equation
\begin{equation}\label{eq:ODEintro}
\mc L(x,E)\Psi(x,E)=0
\,.
\end{equation}
The above equation has two singular points, namely a regular singularity in  $x=0$ and an irregular singularity in
$x=\infty$,  and a natural connection problem arises when one tries to understand the behavior of the solutions globally.
We show that for every representation $V^{(i)}$ there exists a unique solution  $\Psi^{(i)}(x,E)$ to equation
\eqref{eq:ODEintro} which is subdominant for $x \to + \infty$ (by subdominant we mean the solution that goes to zero
the most rapidly). Then, we define the generalized spectral determinant $Q^{(i)}(E;\ell)$ as the coefficient of the
most singular (yet algebraic) term of the expansion of $\Psi^{(i)}(x,E)$ around $x=0$. Finally, we prove that for
generic $\ell \in \mf h$ the generalized spectral determinants $Q^{(i)}(E;\ell)$ are entire functions and satisfy
the Bethe Ansatz equation \eqref{eq:TBAintro}, also known as $Q$-system.

\medskip

The paper is organized as follows. In Section 1 we review some basic facts about the theory of finite dimensional Lie
algebras, affine Kac-Moody algebras and their finite dimensional representations. Furthermore, we introduce the
relevant representations that we will consider throughout the paper.

\medskip

Section 2 is devoted to the asymptotic analysis of equation \eqref{eq:ODEintro}. The main result is provided by
Theorem \ref{thm:asymptotic}, from which the existence of the fundamental (subdominant) solution $\Psi^{(i)}(x,E)$
follows. The asymptotic properties of the solution turn out to depend on the spectrum of
$\Lambda=e_0+e\in\widehat{\mf g}$, as observed in \cite{Sun12}.

\medskip
The main result of Section 3 is Theorem \ref{thm:psi-sistem}, which establishes the following quadratic
relation among the functions $\Psi^{(i)}(x,E)$, known as $\Psi$-system, that was conjectured in \cite{dorey07}:
\begin{equation}\label{eq:PsiIntro}
m_i\big(  \Psi_{-\frac{1}{2}}^{(i)}(x,E) \wedge \Psi_{\frac{1}{2}}^{(i)}(x,E) \big) =\otimes_{j\in I} \Psi^{(j)}(x,E)^{\otimes B_{ij}}
\,.
\end{equation}
Here, $\Psi_{\pm\frac{1}{2}}^{(i)}(x,E)$ is the twisting of the solution $\Psi^{(i)}(x,E)$ defined in
\eqref{eq:Psik}, while $m_i$ is the morphism of $\widehat{\mf{g}}$-modules defined by \eqref{m_i-affine},
and $B=(B_{ij})_{i,j=1}^n$ denotes the incidence matrix of $\mf g$. In order to prove Theorem \ref{thm:psi-sistem},
we first establish in Proposition \ref{prop:pisa2} some important properties of the eigenvalues of $\Lambda$ on
the relevant representations previously introduced. More precisely, we show that in each representation $V^{(i)}$
there exists a maximal positive eigenvalue $\lambda^{(i)}$ (see Definition \ref{def:maximal}), and that the following
remarkable identity holds:
\begin{equation}\label{eq:lambdaintro}
\left(e^{-\frac{\pi i}{h^\vee}}+e^{\frac{\pi i}{h^\vee}}\right)\lambda^{(i)} =\sum_{j=1}^nB_{ij}\lambda^{(j)}\,,
\qquad i=1,\dots,n=\rank\mf g
\,.
\end{equation}
Equation \eqref{eq:lambdaintro} shows that the vector $(\lambda^{(1)},\dots,\lambda^{(n)})$ is the Perron-Frobenius
eigenvector of the incidence matrix $B$, and this implies that the $\lambda^{(i)}\mbox{'s}$ are (proportional to) the
masses of the Affine Toda Field Theory with the same underlying Lie algebra $\mf g$ \cite{olive92}.

\medskip
In Section \ref{sec:Q} we derive the Bethe Ansatz equations \eqref{eq:TBAintro}.
To this aim, we study the local behavior of equations \eqref{eq:ODEintro} close to the Fuchsian
singularity $x=0$, and we define the generalized spectral determinants $Q^{(i)}(E;\ell)$ and $\widetilde{Q}^{(i)}(E;\ell)$
using some properties of the Weyl group of $\mf g$.
In addition, using the $\Psi$-system \eqref{eq:PsiIntro} we prove Theorem \ref{thm:QQtilde}, which gives a set of
quadratic relations among the generalized spectral determinants (equation \eqref{eq:QQtilde}).
Evaluating these relations at the zeros of the functions $Q^{(i)}(E;\ell)$, we obtain
the Bethe Ansatz equations \eqref{eq:TBAintro}.

\medskip
Finally, in Section \ref{app:airy}, we briefly study an integral representation of the subdominant
solution of equation \eqref{eq:ODEintro} with a linear potential $p(x,E)=x$, and we compute
the spectral determinants $Q^{(i)}(E;\ell)$.
Some explicit computations of the eigenvectors of $\Lambda$ are provided in Appendix \ref{app:PF}.
\medskip

In this work we do not study the asymptotic behavior of the $Q^{(i)}(E)$'s, for $E \ll 0$,
that was conjectured in \cite{dorey07}, neither we address the problem of solving the Bethe Ansatz equations
\eqref{eq:TBAintro} -- for a given asymptotic behavior of the functions $Q^{(i)}(E)$'s -- via the Destri-de
Vega integral equation \cite{destri92}. These problems are of a rather different mathematical nature with
respect to the ones treated in this work, and they will be considered in a separate paper.

\subsection*{Note added in press} We proved the ODE/IM correspondence for non simply--laced Lie algebras in \cite{marava15-2}.

\subsection*{Acknowledgments}
The authors were partially supported by the INdAM-GNFM ``Progetto Giovani 2014''. D. M. is supported by the
FCT scholarship, number SFRH/BPD/75908/2011.
D.V. is supported by an NSFC ``Research Fund for International Young Scientists'' grant.
Part of the work was completed during the visit of A.R. and
D.V. to the Grupo de F\'isica Matem\'atica of the Universidade de Lisboa, and during the participation of the
authors to the intensive research program ``Perspectives in Lie Theory'' at the Centro di Ricerca
Matematica E. De Giorgi in Pisa. We wish to thank these institutions for their kind hospitality.
D.M thanks the Dipartimento di Fisica of the Universit\`a degli Studi di Torino for the kind hospitality.
A. R. thanks the Dipartimento di Matematica of the Universit\`a degli Studi di Genova for the kind hospitality.
We also wish to thank Vladimir Bazhanov, Alberto De Sole, Edward Frenkel, Victor Kac and Roberto Tateo for useful discussions.

\medskip

While we were writing the present paper, we became aware of the work \cite{dorey15} (see \cite{negro14}
for a more detailed exposition) that was being written by Patrick Dorey, Stefano Negro and Roberto Tateo.
Their paper focuses on ODE /IM correspondence for massive models related to affine Lie algebras and
deals with similar problems in representation theory. We thank Patrick, Stefano and Roberto for
sharing a preliminary version of their paper.

\section{Affine Kac-Moody algebras and their finite dimensional representations}\label{sec:1}
In this section we review some basic facts about simple Lie algebras, affine Kac-Moody algebras and their
finite dimensional representations which will be used throughout the paper.
The discussion is restricted to simple Lie algebras of $ADE$ type, but most of the results hold also
in the non-simply laced case; we refer to \cite{FH91,Kac90} for further details. At the end of
the section we introduce a class of $\widehat{\mf g}$-valued connections and the associated differential equations.

\subsection{Simple lie algebras and fundamental representations}\label{sec:lie_algebras}

Let $\mf g$ be a simple finite dimensional Lie algebra of $ADE$ type, and let $n$ be its rank.
The Dynkin diagrams associated to these algebras are given in Table \ref{fig:dynkin}.
Let $C=(C_{ij})_{i,j\in I}$ be the Cartan matrix of $\mf g$, and let us denote by
$B=2\mbbm1_n-C$ the
corresponding incidence matrix.
Since $\mf g$ is simply-laced, it follows that  $B_{ij}=1$ if the nodes $i,j$
of the Dynkin diagram of $\mf g$ are linked by an edge, and $B_{ij}=0$ otherwise.
%

\begin{table}[H]
\caption{Dynkin diagrams of simple Lie algebras of $ADE$ type.}
\label{fig:dynkin}
\begin{align*}
&
A_n
\quad
\begin{tikzpicture}[start chain]
\dnode{1}
\dnode{2}
\dydots
\dnode{n-1}
\dnode{n}
\end{tikzpicture}
&\quad\quad&
E_6
\quad
\begin{tikzpicture}
\begin{scope}[start chain]
\foreach \dyni in {1,2,3,5,6} {
\dnode{\dyni}
}
\end{scope}
\begin{scope}[start chain=br going above]
\chainin (chain-3);
\dnodebr{4}
\end{scope}
\end{tikzpicture}
\\
&
D_n
\quad
\begin{tikzpicture}
\begin{scope}[start chain]
\dnode{1}
\dnode{2}
\node[chj,draw=none] {\dots};
\dnode{n-2}
\dnode{n}
\end{scope}
\begin{scope}[start chain=br going above]
\chainin(chain-4);
\dnodebr{n-1}
\end{scope}
\end{tikzpicture}
&\quad\quad&
E_7
\quad
\begin{tikzpicture}
\begin{scope}[start chain]
\foreach \dyni in {1,2,3,4,6,7} {
\dnode{\dyni}
}
\end{scope}
\begin{scope}[start chain=br going above]
\chainin (chain-4);
\dnodebr{5}
\end{scope}
\end{tikzpicture}
\\
&
&\quad\quad&
\!\!\!\!\!\!\!\!\!\!\!\!\!\!\!\!\!\!\!\!\!\!\!\!\!\!\!\!\!\!\!\!
\!\!\!\!\!\!\!\!\!\!\!\!\!\!\!\!\!\!\!\!\!\!\!\!\!\!\!\!\!\!\!\!
E_8
\quad
\begin{tikzpicture}
\begin{scope}[start chain]
\foreach \dyni in {1,2,3,4,5,7,8} {
\dnode{\dyni}
}
\end{scope}
\begin{scope}[start chain=br going above]
\chainin (chain-5);
\dnodebr{6}
\end{scope}
\end{tikzpicture}
&
\end{align*}
\end{table}

\noindent
Let $\{f_i,h_i,e_i\mid  i\in I=\{1,\dots,n\}\}\subset \mf g$
be the set of Chevalley generators of $\mf g$.
They satisfy the following relations ($i,j\in I$):
\begin{equation}\label{eq:chevalley}
[h_i,h_j]=0\,,
\qquad
[h_i,e_j]=C_{ij}e_j\,,
\qquad
[h_i,f_j]=-C_{ij}f_j\,,
\qquad
[e_i,f_j]=\delta_{ij} h_i\,,
\end{equation}
together with the Serre's identities.
Recall that
$$
\mf h=\bigoplus_{i\in I}\mb C h_i\subset\mf g
$$
is a Cartan subalgebra of $\mf g$ with corresponding Cartan decomposition
\begin{equation}\label{eq:cartan_dec}
\mf g=\mf h\oplus\left(\bigoplus_{\alpha\in R}\mf g_\alpha\right).
\end{equation}
As usual, the $\mf g_\alpha $ are the eigenspaces of the adjoint representation
and $R\subset \mf{h}^\ast$ is the set of roots.
For every $i\in I$, let $\alpha_i\in\mf h^*$ be defined by
$\alpha_i(h_j)=C_{ij}$, for every $j\in I$. Then $\mf g_{\alpha_i}=\mb Ce_{\alpha_i}$.
Hence, $\Delta=\{\alpha_i\mid i\in I\}\subset R$.
The elements $\alpha_1,\dots,\alpha_n$ are called the simple roots of $\mf g$.
They can be associated to each vertex of the Dynkin diagram as in Table \ref{fig:dynkin}. We denote by
$$
P=\{\lambda\in\mf h^*\mid \lambda(h_i)\in\mb Z,\,\forall i\in I\}\subset\mf h^*
$$
the set of weights of $\mf g$ and by
$$
P^+=\{\lambda\in P\mid \lambda(h_i)\geq0,\,\forall i\in I\}\subset P
\,
$$
the subset of dominant weights of $\mf g$.
We recall that we can define a partial ordering on $P$ as follows: for $\lambda,\mu\in P$, we say that
$\lambda<\mu$ if $\lambda-\mu$ is a sum of positive roots.  For every $\lambda\in\mf h^*$ there exists a
unique irreducible representation $L(\lambda)$ of
$\mf g$ such that there is $v\in L(\lambda)\setminus\{0\}$ satisfying
$$
h_iv=\lambda(h_i)v\,,
\qquad
e_iv=0\,,
\qquad
\text{for every }i\in I
\,.
$$
$L(\lambda)$ is called the highest weight representation of weight $\lambda$, and $v$ is the highest weight
vector of the representation. We have that $L(\lambda)$ is finite dimensional if and only if
$\lambda\in P^+$, and conversely, any irreducible finite dimensional representation of $\mf g$ is of the form
$L(\lambda)$ for some $\lambda\in P^+$. For $\lambda\in P^+$, the representation $L(\lambda)$ can be decomposed in the direct sum of its
(finite dimensional) weight spaces $L(\lambda)_\mu$, with $\mu\in P$,
and we have that $L(\lambda)_\mu\neq0$ if and only if $\mu<\lambda$. We denote by $P_\lambda$ the set of weights
appearing in the weight space decomposition of $L(\lambda)$; the multiplicity of $\mu$ in the representation
$L(\lambda)$ is defined as the dimension of the weight space $L(\lambda)_\mu$.
\\

\noindent
Recall that the fundamental weights of $\mf g$ are those elements $\omega_i\in P^+,\, i\in I$,  satisfying
\begin{equation}\label{20150107:eq1}
\omega_i(h_j)=\delta_{ij}\,,
\qquad
\text{for every }j\in I
\,.
\end{equation}
The corresponding highest weight representations  $L(\omega_i)$, $i\in I$,  are known as fundamental
representations of $\mf g$, and for every $i\in I$  we denote by $v_i\in L(\omega_i)$ the highest weight
vector of the representation $L(\omega_i)$. 
 Since the simple roots and the
fundamental weights of $\mf g$ are related via the Cartan matrix (which is non-degenerate) by the relation
$$
\omega_i=\sum_{j\in I}(C^{-1})_{ji}\alpha_j, \qquad i\in I
\,,
$$
then we can associate to the $i-$th vertex of the Dynkin diagram of $\mf g$
the corresponding fundamental representation $L(\omega_i)$ of $\mf g$.

\subsection{The representations
\texorpdfstring{$L(\eta_i)$}{Leta_i}}\label{sec:Wtilde}
Let us consider the dominant weights 
$$
\eta_i=\sum_{j\in I}B_{ij}\omega_j, \qquad i\in I
\,,
$$
where $B=(B_{ij})_{i,j\in I}$ is the incidence matrix defined in Section \ref{sec:lie_algebras}, and let
$L(\eta_i)$ be the corresponding irreducible finite dimensional representations.
In addition, we consider the tensor product representations
$$
\bigotimes_{j\in I}L(\omega_j)^{\otimes B_{ij}}\,,
\qquad i\in I
\,.
$$
We now show that we can find a copy of the irreducible representation
$L(\eta_i)$ inside the representations $\bigwedge^2L(\omega_i)$ and
$\bigotimes_{j\in I}L(\omega_j)^{\otimes B_{ij}}$.
This will be used in Section \ref{sec:psi-sistem} to construct the so-called $\psi$-system. %
\begin{lemma}\label{lem:021202}
The following facts hold in the representation $\bigwedge^2 L(\omega_i)$.
\begin{enumerate}[(a)]
\item For every $j\in I$, we have that
$$
e_j(f_iv_i \wedge v_i)=0
\,.
$$
\item For every $h \in \mf h$, we have that
$$
h(f_iv_i \wedge v_i)=\eta_i(h)(f_iv_i\wedge v_i)
\,.
$$
\item Any other weight of the representation $\bigwedge^2
L(\omega_i)$ is smaller than $\eta_i$.
\end{enumerate}
\end{lemma}
\begin{proof}
Since $v_i$ is a highest weight vector for $L(\omega_i)$ we have that
$e_jv_i=0$ for every $j\in I$.
Moreover, using the relations $[e_j,f_i]=\delta_{ij}h_i$ together with equation \eqref{20150107:eq1},
we have
$$
e_j f_i v_i= f_i e_jv_i+\delta_{ij}h_i v_i=\delta_{ij}v_i
\,,
$$
and therefore $e_j(f_iv_i\wedge v_i)=\delta_{ij}v_i\wedge v_i=0$, proving part (a). By linearity, it
suffices to show part (b) for $h=h_i$, $i\in I$.
Recall that, by equation \eqref{eq:chevalley}, we have $[h_i,f_j]=-C_{ij}f_j$, for every $i,j\in I$. Then,
$$
h_j f_i v_i= \left(f_i h_j +[h_j,f_i]\right)v_i
=(\delta_{ij}- C_{ji}) f_i v_i
\,.
$$
Hence, $h_j (f_i v_i \wedge v_i) = (2 \delta_{ij} -C_{ji}) (f_i v_i \wedge v_i)
=\eta_i(h_j)(f_iv_i\wedge v_i)$. Finally, let $w\neq \omega_i,\eta_i-\omega_i$ be a weight appearing in $L(\omega_i)$.
It follows from representation theory of simple Lie algebras that
$$
\omega<\eta_i-\omega_i=\omega_i-\sum_{j\in I}C_{ji}\omega_j<\omega_i
\,.
$$
Therefore, by part (b), the maximal weight of $\bigwedge^2 L(\omega_i)$ is
$w_i+\eta_i-\omega_i=\eta_i$, thus proving part (c).
\end{proof}
\noindent
As a consequence of the above Lemma it follows that $f_iv_i\wedge v_i\in\bigwedge^2L(\omega_i)$
is a highest weight vector of weight $\eta_i$.
Since $\mf g$ is simple, the subrepresentation of $\bigwedge^2L(\omega_i)$
generated by the highest weight vector $f_iv_i\wedge v_i$ is irreducible.
Therefore, it is isomorphic to $L(\eta_i)$.
By the complete reducibility of $\bigwedge^2L(\omega_i)$ we can decompose it as follows:
$$
\bigwedge^2L(\omega_i)=L(\eta_i)\oplus U\,,
$$
where $U$ is the direct sum of all the irreducible representations different from $L(\eta_i)$. By an
abuse of notation we denote with the same symbol the representation $L(\eta_i)$ and its copy in
$\bigwedge^2L(\omega_i)$. For every $i\in I$, let us denote by $w_i=\otimes_{j\in I} v_j^{\otimes B_{ij}}$.
The following analogue of Lemma \ref{lem:021202} in the case of  the representation
$\bigotimes_{j\in I}L(\omega_j)^{\otimes B_{ij}}$ holds true.
\begin{lemma}\label{lem:231202}
The following facts hold in the representation $\bigotimes_{j\in I}L(\omega_j)^{\otimes B_{ij}}$.
\begin{enumerate}[(a)]
\item For every $j\in I$, we have that
$$
e_jw_i=0
\,.
$$
\item For every $h \in \mf h$, we have that
$$
h w_i=\eta_i(h)w_i
\,.
$$
\item Any other weight of the representation $\bigotimes_{j\in I}L(\omega_j)^{\otimes B_{ij}}$
is smaller than $\eta_i$.
\end{enumerate}
\end{lemma}
\begin{proof}
Same as the proof of Lemma \ref{lem:021202}.
\end{proof}
As for the previous discussion, by Lemma \ref{lem:231202} it follows
that $w_i\in \bigotimes_{j\in I}L(\omega_j)^{\otimes B_{ij}}$
is a highest weight vector which generates an irreducible subrepresentation isomorphic to $L(\eta_i)$.
Lemmas \ref{lem:021202} and \ref{lem:231202}, together with the Schur Lemma, imply that there exists a
unique morphism of representations
\begin{equation}\label{m_i}
m_i: \bigwedge^2 L(\omega_i) \to
\bigotimes_{j\in I}L(\omega_j)^{\otimes B_{ij}}
\,,
\end{equation}
such that $\ker m_i= U$ and $m_i( f_iv_i \wedge v_i )=w_i$.
%

\subsection{Affine Kac-Moody algebras and finite dimensional representations}\label{sec:kac_moody}

Let $h^\vee$ be the dual Coxeter number of $\mf g$.
Let us denote by $\kappa$ the Killing form of $\mf g$ and let us fix the following
non-degenerate symmetric invariant bilinear form on $\mf g$ ($a,b\in\mf g$):
$$
(a| b)=\frac{1}{h^\vee}\kappa(a|b)
\,.
$$
Let $\mf g\otimes\mb C[t,t^{-1}]$ be the space of Laurent polynomials with coefficients in $\mf g$.
For $a(t)=\sum_{i=-M}^Na_it^i\in\mb C[t,t^{-1}]$, we let
$$
\res_{t=0}a(t)dt=a_{-1}
\,.
$$
The affine Kac-Moody algebra $\widehat{\mf g}$ is the vector space
$\widehat{\mf g}=\mf g\otimes\mb C[t,t^{-1}]\oplus\mb Cc$
endowed with the following Lie algebra structure ($a,b\in\mf g,f(t),g(t)\in\mb C[t,t^{-1}]$):
\begin{align}
\begin{split}
\label{20141020:eq3}
&[a\otimes f(t),b\otimes g(t)]=[a,b]\otimes f(t)g(t)+(a| b)\res_{t=0}\left(f'(t)g(t)dt\right)c\,,
\\
&[c,\widehat{\mf g}]=0
\,.
\end{split}
\end{align}
The set of Chevalley generators of $\widehat{\mf g}$ is obtained by adding to the Chevalley
generators of $\mf g$ some new generators $f_0,h_0,e_0$ (for the construction, see
for example \cite{Kac90}). The generator $e_0$, which plays an important role in the paper,
can be constructed as follows. There exists a root $-\theta\in R$, known as
the lowest root of $\mf g$, such that $-\theta-\alpha_i\not\in R$, for every
$i\in I$.  Then
\begin{equation}\label{20151129:eq1}
e_0=a\otimes t
\,,
\qquad\text{for some }a\in\mf g_{-\theta}
\,.
\end{equation}
\noindent
We now consider an important class of representations of $\widehat{\mf g}$. Let $V$ be a finite dimensional representation of $\mf g$.
For $\zeta\in\mb C^*$ we define a representation of $\widehat{\mf g}$, which we denote
by $V(\zeta)$, as follows:
as a vector space we take $V(\zeta)= V$, and the action of $\widehat{\mf g}$ is defined 
by
$$
(a \otimes f(t))v=f(\zeta)(av)\,,
\qquad
c\,v=0\,,
\qquad
\text{for }a\in\mf g\,,f(t)\in\mb C[t,t^{-1}]\,, v\in V
\,.
$$
The representation $V(\zeta)$ is called an evaluation representation of $\widehat{\mf g}$.
If $V$ and $W$ are representations of $\mf g$ then any morphism of representations
$f:V\to W$ can obviously be extended to a morphism of representations $f:V(\zeta)\to W(\zeta)$,
which we denote by the same letter by an abuse of notation. 
Similarly, when referring to weights or weight vectors of the evaluation representation
$V(\zeta)$, we mean the weights and weight vectors of the representation $V$ with respect to the action 
of $\mf g$.
\\

For a representation $V$ of $\mf g$ and  $k\in\mb C$, we denote $V_k=V(e^{2\pi i k})$ 
the evaluation representation of $\widehat{\mf g}$ corresponding to $\zeta=e^{2\pi i k}$.
Clearly, if $k\in\mb Z$ then $V_k=V_0=V(1)$, and if $k\in\frac12+\mb Z$ then
$V_k=V_{\frac12}=V(-1)$.
Then, for each vertex $i\in I$ of the Dynkin diagram of $\mf g$, we associate the following
finite dimensional representation of $\widehat{\mf g}$:
\begin{equation}\label{fund_rep}
V^{(i)}=L(\omega_i)_{\frac{p(i)}2}
\,,
\end{equation}
where $L(\omega_i)$ is the fundamental representation of $\mf g$ with highest weight $\omega_i$, and the
function $p:I\to\mb Z/2\mb Z$ is defined inductively as follows:
\begin{equation}\label{eq:p}
p(1)=0\,,
\qquad
p(i)=p(j)+1\,,
\text{ for }j<i\text{ such that } B_{ij}>0.
\end{equation}
 By an abuse of language we call the representations \eqref{fund_rep}, for $i\in I$, the finite dimensional
 fundamental representations of $\widehat{\mf g}$. These representations will play a fundamental role in the present paper.

\begin{example}\label{exa:An}
In type $A_n$,  $V^{(1)}=L(\omega_1)_0$ is the evaluation representation at $\zeta=1$
of the standard representation $L(\omega_1)=\bb{C}^{n+1}$.
Moreover,
$$
V^{(i)}=\bigwedge^iV_{\frac{i-1}{2}}^{(1)}
\,,
\qquad
\text{for }k=1,\dots,n
\,.
$$
\end{example}
\begin{example}\label{exa:Dn}
In type $D_n$,  $V^{(1)}=L(\omega_1)_0$ is the evaluation representation at $\zeta=1$
of the standard representation $L(\omega_1)=\bb{C}^{2n}$.
Moreover,
$$
V^{(i)}=\bigwedge^iV_{\frac{i-1}{2}}^{(1)}
\,,
\qquad
\text{for }k=1,\dots,n-2
\,.
$$
Finally, $V^{(n-1)}=L(\omega_{n-1})_{\frac n2}$ and
$V^{(n)}=L(\omega_{n})_{\frac n2}$ are the evaluation representations at $\zeta=(-1)^n$
of the so-called half-spin representations.
\end{example}
\begin{example}\label{exa:E6}
In type $E_6$, $V^{(1)}=L(\omega_{1})_{0}$ and $V^{(5)}=L(\omega_{5})_{0}$ are the evaluation
at $\zeta=1$ of the two $27$-dimensional representations (they are dual to each other).
Moreover,
$$V^{(2)}=\bigwedge^2V_{\frac{1}{2}}^{(1)}\,,
\qquad
V^{(3)}=\bigwedge^3V_{0}^{(1)}\cong\bigwedge^3V_{0}^{(5)}\,,
\qquad
V^{(4)}=\bigwedge^2V_{\frac{1}{2}}^{(5)}
\,.
$$
Finally, $V^{(6)}=L(\omega_{6})_{\frac 12}$ is the evaluation representation at $\zeta=-1$
of the adjoint representation.
\end{example}

\noindent
As a final remark, note that by equation \eqref{fund_rep} we have
\begin{equation}\label{20150108:eq2}
\left(\bigwedge^2L(\omega_i)\right)_{\frac{p(i)+1}{2}}
=\bigwedge^2V^{(i)}_{\frac{1}{2}}
\,,
\end{equation}
and if we introduce the evaluation representation
\begin{equation}\label{20150108:eq1}
M^{(i)}=\bigotimes_{j\in I}\big(V^{(j)}\big)^{\otimes B_{ij}}, \qquad \forall i\in I
\,,
\end{equation}
then by equation \eqref{fund_rep}, and from the fact that  $B_{ij}\neq0$ implies $p(i)=p(j)+1$, it follows that
the morphism $m_i$ defined by equation \eqref{m_i} extends to a morphism of
evaluation representations
\begin{equation}\label{m_i-affine}
m_i:\bigwedge^2V^{(i)}_{\frac{1}{2}}\to M^{(i)}
\,.
\end{equation}
The construction of the $\psi-$system will be provided by means of the above extended morphism.
Due to Lemmas \ref{lem:021202} and \ref{lem:231202}, the evaluation representation
\begin{equation}\label{eq:evWi}
W^{(i)}=L(\eta_i)_{\frac{p(i)+1}{2}} \,, i\in I\,,
\end{equation}
is a subrepresentation of both $\bigwedge^2V^{(i)}_{\frac{1}{2}}$ and $M^{(i)}$, and the copies of $W^{(i)}$
in these representations are identified by means of the morphism $m_i$.

\subsection{\texorpdfstring{$\widehat{\mf g}$}{ghat}-valued connections
and differential equations}\label{sec:L}

Let $\{f_i, h_i,e_i\mid i=0,\dots, n\}\subset\widehat{\mf g}$
be the set of Chevalley generators of $\widehat{\mf g}$,
and let us denote $e=\sum_{i=1}^ne_i$.
Fix an element $\ell\in\mf h$.
Following \cite{FF11} (see also \cite{Sun12}),
we consider the $\widehat{\mf g}$-valued connection 
\begin{equation}\label{20141020:eq1}
\mc L(x,E)=\partial_x+\frac{\ell}{x}+e+p(x,E) e_0\,,
\end{equation}
where $p(x,E)=x^{Mh^\vee}-E$, with $M>0$ and $E\in\mb C$.
Since $e$ is a principal nilpotent element, it follows from the Jacobson-Morozov Theorem together with equation
\eqref{20151129:eq1} that there exists an element $h\in\mf h$ such that 
\begin{equation}\label{20141020:eq2}
[h,e]=e
\,,
\qquad
[h,e_0]=-(h^\vee-1)e_0
\,,
\end{equation}
see for instance \cite{CMG93}. Under the isomorphism between $\mf{h}$ and $\mf{h}^\ast$ provided by the Killing form, the element $h$ corresponds to the Weyl vector (half sum of positive roots).  Let $k\in\mb C$  and introduce the quantities
$$
\omega=e^{\frac{2\pi i}{h^\vee(M+1)}}
\,,
\qquad
\qquad
\Omega=e^{\frac{2\pi iM}{M+1}}=\omega^{h^\vee M}
\,.
$$
Moreover, let $\mc{M}_{k}$ be the automorphism of $\widehat{\mf g}$-valued connections which
fixes $\partial_x$, $\mf g$ and $c$ and sends $t \to e^{2\pi i k} t$.
Then from equations \eqref{20141020:eq3}, \eqref{20141020:eq1} and \eqref{20141020:eq2}
we get 
\begin{equation}\label{20141020:eq4}
\mc{M}_{k}\left(\omega^{k\ad h}\mc L(x,E)\right)
=\omega^k \mc L(\omega^k x,\Omega^kE)
\,.
\end{equation}
We denote $\mc L_k(x,E)=\mc M_k\left(\mc L(x,E)\right)$, for every $k\in\mb C$. Note in particular
that $\mc L_0(x,E)=\mc L(x,E)$. Equation \eqref{20141020:eq4}  implies that for a given  finite
dimensional representation $V$  of $\mf g$, the action of $\mc L_k(x,E)$ on the evaluation representation $V_0$
of $\widehat{\mf g}$ is the same as the action
of $\mc L(x,E)$ on $V_k$, for every $k\in\mb C$.
In addition, let $ \widehat{\mb{C}}$ be the universal cover of $\bb{C}^{*}$,  suppose $\phi(x,E):\widehat{\mb{C}} \to V_0$ 
is a family -- depending on the parameter  $E$ -- of solutions of the (system of) ODE
\begin{equation}\label{20141021:eq1}
\mc L(x,E)\phi(x,E)=0
\,,
\end{equation}
and for $k\in\mb C$ introduce the function
\begin{equation}\label{20150108:eq8}
\phi_k(x,E)=\omega^{-kh}\phi(\omega^kx,\Omega^kE)
\,.
\end{equation}
Since on any evaluation representation of $\widehat{\mf g}$ we have 
\begin{equation}\label{20141230:eq1}
e^{\ad a}b=e^a b e^{-a}\,,
\qquad a,b\in\widehat{\mf g}
\,,
\end{equation}
then by equations \eqref{20141020:eq4} and \eqref{20141230:eq1}, we have that
\beq\label{20141128:eq1}
\mc L_k(x,E)\phi_k(x,E)=0
\,,
\eeq
namely $\phi_k(x,E):\widehat{\mb{C}} \to V_k$ is a solution of \eqref{20141021:eq1} for the representation
$V_k$.

\section{Asymptotic Analysis}
Let $V$ be an evaluation representation
of $\widehat{\mf g}$.
Let us denote by
$\widetilde{\mb{C}}=\mb C\setminus\mb R_{\leq0}$, the complex plane minus the semi-axis of real
non-positive numbers,
and let us consider a solution $\Psi:\widetilde{\mb{C}}\rightarrow V$ of
\beq\label{20141125:eq1}
\mc{L}(x,E)\Psi(x)=\Psi'(x)+ \left(\frac{\ell}{x}+e+p(x,E) e_0\right) \Psi(x) =0
\,.
\eeq
Equation \eqref{20141125:eq1} is a (system of) linear ODEs, with a Fuchsian singularity at $x=0$
and an irregular singularity at $x=\infty$. Our goal is to prove the existence of solutions of equation \eqref{20141125:eq1}
with a prescribed asymptotic behavior at infinity in a Stokes sector of the complex plane.
In the present form, however, equation \eqref{20141125:eq1}
falls outside the reach of classical results such as the
Levinson theorem \cite{Lev48} or the complex Wentzel-Kramers-Brillouin
(WKB) method \cite{fedoryuk93},
since the most singular term $p(x,E)e_0$ is nilpotent.
In order to study the asymptotic behavior of solutions to \eqref{20141125:eq1}
we use a slight modification of a gauge transform originally introduced in \cite{Sun12},
and we then prove the existence of the desired solution
for the new ODE by means of a Volterra integral equation.

We now consider the required transformation. First, we note that the function $p(x,E)^\frac{1}{h^{\vee}}$ 
has the asymptotic expansion
$$p(x,E)^\frac{1}{h^{\vee}}=q(x,E)+O(x^{-1-\delta}),$$
where
\begin{equation}\label{eq:delta}
\delta=M(h^\vee(1+s)-1)-1>0, \qquad s=\lfloor \frac{M+1}{h^\vee M} \rfloor,
\end{equation}
and
\begin{equation}\label{20150128eq1}
q(x,E)=x^M+\sum_{j=1}^{s} c_j(E) x^{M(1-h^\vee j)}.
\end{equation}
For every $j=1,\dots,s$, the function $c_j(E)$ is a monomial of degree $j$ in $E$.
\begin{lemma}\label{20150103:lem1}
Let $\mc L(x,E)$ be the $\widehat{\mf g}$-valued connection defined in \eqref{20141020:eq1} and
let $h\in\mf{h}$ satisfy relations \eqref{20141020:eq2}.
Then, we have the following gauge transformation
\begin{equation}\label{eq:24ott03}
q(x,E)^{\ad h}\mc{L}(x,E)
=\partial_x+ q(x,E) \Lambda + \frac{\ell- M h}{x}+O(x^{-1-\delta})
\,,
\end{equation}
where $\Lambda=e_0+e$, the function $q(x,E)$ is given by \eqref{20150128eq1} and $\delta$ by \eqref{eq:delta}.
\end{lemma}
\begin{proof}
It follows by a straightforward computation using relations \eqref{20141020:eq2}.
\end{proof}


%
\begin{remark}
The transformation \eqref{eq:24ott03} differs from the one used in \cite{Sun12} and given by
\begin{equation*}\label{eq:24ott02}
x^{M\ad h}\mc{L}(x,E)
=\partial_x+  x^M \Lambda + \frac{\ell- M h}{x} - \frac{e_0 E}{x^{M(h^\vee-1)}}
\,.
\end{equation*}
The latter transformation fails to be the correct if $M (h^\vee-1)\leq 1 $. In fact, in this case
the term $ e_0 E/x^{M(h^\vee-1)}$ goes to zero slowly and alters the asymptotic behavior.
\end{remark}
We are then left to analyze the equation
\begin{equation}\label{eq:21ott01}
\Psi'(x)+ \left( q(x,E)\, \Lambda  + \frac{\ell-Mh}{x}+ A(x) \right) \Psi(x) =0
\,,
\end{equation}
where $A(x)=O(x^{-1-M h^\vee})$ is a certain matrix-valued function. From the general theory we expect the asymptotic behavior to depend on the primitive of $q(x,E)$, which is known as the action. In the generic case $\frac{M+1}{h^\vee M} \notin \bb{Z}_+$, the action is defined as
\begin{equation}\label{eq:actiongeneric}
S(x,E)=\int^{x}_0 q(y,E) dy \,,
\qquad x \in \widetilde{\mb C}
\,,
\end{equation}
where
we chose the branch
of $q(x,E)$ satisfying $q \sim |x|^{M}$ for $x$ real. In the case $\frac{M+1}{h^\vee M} \in \bb{Z}_+$ then formula \eqref{eq:actiongeneric} becomes
\begin{equation*}
 S(x,E)=\sum_{j=0}^{s-1} \int^{x}_0   c_j(E) y^{M(1-h^\vee j)} dy + c_s(E) \log x \, , \quad s= \frac{M+1}{h^\vee M} \,.
\end{equation*}
Existence and uniqueness of solutions to equation \eqref{20141125:eq1} (or equivalently, of \eqref{eq:21ott01})
depend on the properties of the element $\Lambda=e+e_0$ in the representation considered. We therefore introduce the following
\begin{definition}\label{def:maximal}
Let $A$ be an endomorphism of a vector space $V$.
We say that a eigenvalue $\lambda$ of $A$ is maximal if it is real,  its algebraic multiplicity
is one, and $\lambda > \Re\mu$ for every eigenvalue $\mu$ of $A$.
\end{definition}
\noindent
We are now in the position to state the main result of this section.%
\begin{theorem}\label{thm:asymptotic}
Let $V$ be an evaluation representation of $\widehat{\mf g}$, and let the matrix representing
$\Lambda\in\widehat{\mf g}$ in $V$ have a maximal eigenvalue
$\lambda$. Let  $\psi\in V$ be the corresponding unique (up to a constant) eigenvector. Then,
there exists a unique solution $\Psi(x,E):\widetilde{\mb{C}}\to V$
to equation \eqref{20141125:eq1}
with the following asymptotic behavior:
$$
\Psi(x,E)
=e^{-\lambda S(x,E)} q(x,E)^{-h} \big( \psi + o(1) \big)\quad \text{ as }\quad x \to +\infty
\,.
$$
Moreover, the same asymptotic behavior holds in the sector
$|\arg{x}| < \frac{\pi}{2(M+1)} $,
that is, for any $\delta>0$ it satisfies
\begin{equation}\label{20150113:eq1}
\Psi(x,E)
=e^{-\lambda S(x,E)} q(x,E)^{-h} \big( \psi + o(1) \big)
\,,
\quad \text{ in the sector } \,\, |\arg{x}| <\frac{\pi}{2(M+1)} -\delta
\,.
\end{equation}
The function $\Psi(x,E)$ is an entire function of $E$.
\end{theorem}
\begin{remark}\label{20150107:rem1}
We will prove in Section \ref{sec:psi-sistem} that the hypothesis of the theorem are satisfied
in the $ADE$ case for the representations $V^{(i)}$, $\bigwedge^2V_{\frac12}^{(i)}$, $M^{(i)}$
and $W^{(i)}$, which were introduced in the previous section.
\end{remark}
\begin{remark}
In the case $M(h^\vee-1)>1$, the asymptotic expansion of the solution $\Psi$ in the above theorem reduces to
\begin{equation*}
\Psi(x)
=e^{-\lambda \frac{x^{M+1}}{M+1}} q(x,E)^{-h}\big( \psi + o(1) \big)
\,,
\text{ in the sector } |\arg{x}| <\frac{\pi}{2(M+1)}
\,.
\end{equation*}
\end{remark}
\noindent
Before proving Theorem \ref{thm:asymptotic}, we apply it to  prove the existence of linearly
independent solutions to equation \eqref{20141125:eq1}.
First, note that for any $k\in\mb R$ such that $|k| < \frac{h^\vee (M+1)}{2}$, the function
\begin{equation}\label{eq:Psik}
 \Psi_k(x,E)=\omega^{-k h} \Psi(\omega^kx,\Omega^kE)\,,
 \quad x \in \bb{R}_+
\end{equation}
 defines, by analytic continuation, a solution $\Psi_k: \widetilde{\mb{C}} \to V_k $
 of equation \eqref{20141125:eq1}
 in the representation $V_k$. Theorem \ref{thm:asymptotic} implies the following result.
\begin{corollary}\label{cor:asymptotic}
For any $k\in\mb R$ such that $|k| < \frac{h^\vee (M+1)}{2}$, on the positive real semi-axis the function $\Psi_k$ has the asymptotic behaviour
\begin{equation}\label{eq:asymppsik}
\Psi_k(x,E)=e^{-\lambda \gamma^k S(x,E)}q(x,E)^{-h} \gamma^{-k h}(\psi+o(1))\,, \qquad x \gg 0 \,,
\end{equation}
where $\psi$ is defined as in Theorem \ref{thm:asymptotic} and $\gamma=e^{\frac{2 \pi i}{h^\vee}}$.
Moreover, let $l \in \bb{Z}_+$ be such that $1\leq l \leq \frac{h^\vee}{2}$.
Then, the functions $\Psi_{\frac{1-l}{2}},\dots,\Psi_{\frac{l-1}{2}}$
are linearly independent solutions to equation \eqref{20141125:eq1}
in the representation $V_{\frac{l+1}{2}}$.
\end{corollary}
\begin{proof}
A simple computation shows that $S(\omega^kx,\Omega^kE)=\gamma^k S(x,E)$, and similarly $q(\omega^k x, \Omega^k E)=\omega^Mq(x,E)$.
Therefore by Theorem \ref{thm:asymptotic},
$$
\Psi(\omega^kx,\Omega^kE)=\omega^{-M h}e^{- \gamma^k\lambda   S(x,E) }q(x,E)^{-h}( \psi +o(1))\,,
\qquad x \gg 0 \,.
$$
Hence, on the positive real semi axis we have 
$$
\Psi_k(x,E)=e^{- \gamma^k\lambda   S(x,E) } q(x,E)^{-h} \gamma^{-k h}( \psi +o(1))\,,
\qquad x \gg 0 \,,
$$
where we have used the identity $\omega^{M+1}=\gamma$.
Note that $\gamma^{k}$, $k\in \lbrace \frac{1-l}{2},\dots,\frac{l-1}{2} \rbrace$,
are pairwise distinct.
Therefore, the vectors $\psi_k$ are linearly independent. This concludes the proof.
\end{proof}

\begin{remark}\label{rem:3sectors}
A slight variation of the proof of Theorem \ref{thm:asymptotic} presented below allows one to prove
that the asymptotic behavior of $\Psi$ holds on a larger sector
of the complex plane, consisting of three adjacent Stokes sectors.
This refined result is necessary to consider the lateral connection problem, where one considers
the linear relations among solutions that are subdominant
in different Stokes sectors. Such a problem leads naturally to another instance of the ODE/IM correspondence,
namely to spectral determinants related to the transfer matrix of integrable systems, see e.g. \cite{dtba}.
\end{remark}

\subsection{Proof of Theorem \ref{thm:asymptotic}}\label{sec:thm}
To complete our analysis, we need to further simplify  equation \eqref{eq:21ott01}
with a new gauge transformation.
Let us denote by $\widetilde{\mc L}(x,E)=q(x,E)^{\ad h}\mc L(x,E)$, where $\mc L(x,E)$ and
$h$ are as in Lemma \ref{20150103:lem1}.

\begin{lemma}\label{20150103:lem2}
Let $\ell=\sum_{i\in I}\ell_ih_i$ and $h=\sum_{i\in I}a_ih_i$, with $\ell_i,a_i\in\mb C$,
and consider the element $N=\sum_{i\in I}(\ell_i-Ma_i)f_i\in\mf g\,.$
Then we have the following gauge transformation:
$$
e^{\alpha(x)\ad N}\widetilde{\mc L}(x,E)
=\partial + q(x,E)\,\Lambda+O(x^{-1-M})\,,
$$
where $\alpha(x)=\left(x\,q(x,E)\right)^{-1}\,.$
\end{lemma}
\begin{proof}
It follows by a straightforward computation using the explicit form of $\widetilde{\mc L}(x,E)$ given in
equation \eqref{eq:24ott03}, together with the definitions of $N$ and $\alpha(x)$ and the commutation relations (see \eqref{eq:chevalley})
$$
[f_i, e_0]=0
\qquad\text{and}\qquad
[f_i,e_j]=\delta_{ij}h_i\,,
\qquad\text{for every }i,j\in I\,.
$$
\end{proof}

\noindent
By Lemma \ref{20150103:lem2} and equation \eqref{20141230:eq1}, it follows that the ODE \eqref{eq:21ott01} is transformed into
\begin{equation}\label{eq:22ott01}
\widetilde{\Psi}'(x)+ \left(q(x,E)\, \Lambda  + \widetilde{A}(x) \right) \widetilde{\Psi}(x) =0
\,,
\end{equation}
where $\widetilde{A}(x)=O(x^{-1-M})$,
and  $\widetilde{\Psi}(x)=G (x)\Psi(x)$,
with $G(x)=e^{\alpha(x)N}$.
Since $G(x)=\mbbm1+o(x^{-M})$,
then we look for a solution to equation \eqref{eq:22ott01} with the same
asymptotic behavior \eqref{20150113:eq1} as the solution $\Psi(x)$ to equation \eqref{eq:21ott01}.
Finally, we define $\Phi(x)=e^{\lambda S(x,E)} \widetilde{\Psi}(x)$.
Then, it is easy to show that equation \eqref{eq:22ott01} takes the  form
\begin{equation}\label{eq:21ott02}
\Phi'(x)+ \left( q(x,E)\,\tilde{\Lambda}+ \widetilde{A}(x) \right) \Phi(x) =0
\,,
\end{equation}
where $\tilde{\Lambda}=\Lambda-\lambda \mbbm1$. Note that $\tilde{\Lambda}\,\psi=0$.
The problem is now reduced to prove the existence
of a solution to equation \eqref{eq:21ott02} satisfying, for every $\e >0$, the asymptotic condition
\begin{equation}\label{20141129:eqa1}
\Phi(x)= \psi + o(1) \, ,\quad \text{ in the sector } |\arg{x}| <\frac{\pi}{2(M+1)} -\e
\,.
\end{equation}

To proceed with the proof we need to introduce some elements of WKB analysis, for which we refer to \cite[Chap. 3]{fedoryuk93}.
We fix $\kappa>0$ (possibly depending on $E$)  big enough such that the
Stokes sector 
\begin{equation*}
 \Sigma= \lbrace x \in\widetilde{\mb C}\mid\Re S(x,E) \geq  \Re S(\kappa,E) \rbrace
\end{equation*}
has the following two properties:
$S$ is a bi-holomorphic map from $\Sigma$
to the right half-plane  $\lbrace z \in \bb{C}\mid \Re z \geq 0 \rbrace $, and 
$\Sigma$ coincides asymptotically with the sector $|\arg x|< \frac{\pi}{2(M+1)}$. In other words,
any ray of argument $|\arg x|< \frac{\pi}{2(M+1)}$ eventually lies inside $\Sigma$, while any ray
of argument $|\arg x|> \frac{\pi}{2(M+1)}$ eventually lies in the complement of $\Sigma$. For any
$\e>0$, we define $D_{\e}=S^{-1}\lbrace |\arg z |\leq \frac{\pi}{2} -\e \rbrace $, and we introduce
the space $\mc{B}$  of holomorphic bounded functions $U:D_{\e} \to V$ for which the limit
$\lim_{x \to \infty} U(x)=U(\infty)$ is well-defined. The space $\mc{B}$ is complete with respect
to the norm $\|U\|_{\infty}=\sup_{x \in D_{\e}} |U(x)|$.
\\

\noindent
We construct the solution of equation \eqref{eq:21ott02}
with the desired asymptotic behavior \eqref{20141129:eqa1}
as the solution of the following integral equation in the Banach space $\mc B$:
\begin{equation}\label{eq:21ott05}
\Phi(x)= K[\Phi](x)+ \psi^{(i)}\,,
\end{equation}
where the Volterra integral operator $K$ is defined as
$$
K [\Phi] (x)
= \int_{x}^{\infty}e^{-\tilde{\Lambda}\big( S(x,E)- S(y,E) \big) } \widetilde{A}(y) \Phi(y) dy\,.
$$
The integral above is computed along any admissible path $\mathit{c}$, and we say that a
(piecewise differentiable) parametric curve $\mathit{c}$ is admissible
if $\int_0^t|\dot{\mathit{c}}(s)|ds=O(|\mathit{c}(t)|)$ and $\Re S(\mathit{c}(t),E)$ is a
non decreasing function. For example,  $S^{-1} (t+x) $ is an admissible path \cite{fedoryuk93}.
Note that a solution $\Phi(x)$ of equation \eqref{eq:21ott02}, with the asymptotic behavior
\eqref{20141129:eqa1} belongs to $\mc{B}$.  Moreover, it satisfies equation \eqref{eq:22ott01}
if and only if it satisfies the Volterra integral equation \eqref{eq:21ott05}.
Indeed, provided that $K$ is a continuous operator, we have that
$\lim_{x\to \infty}K[\Phi](x)=0$ and
\begin{align*}
& K[\Phi]'(x)= - \widetilde{A}(x)\Phi(x) -  q(x,E)\, \tilde{\Lambda}  K[\Phi](x)\\
&= - \widetilde{A}(x)\Phi(x) -   q(x,E)\, \tilde{\Lambda} \left(\Phi(x)-\psi^{(i)}\right)\\
&= - \widetilde{A}(x)\Phi(x) - q(x,E)\, \tilde{\Lambda} \Phi(x)
 \, ,
\end{align*}
where the last equality follows from the fact that $\tilde{\Lambda}\psi^{(i)}=0$.
\\

\noindent
In order to prove that the solution of the integral equation \eqref{eq:21ott05} exists and
it is unique we start showing that $K$ is a continuous operator on $\mc{B}$
and its spectral radius is zero, i.e. $\lim_{n \to \infty} \|K^{n}\|^{\frac1n}=0$. By hypothesis,
all eigenvalues of $\Lambda$ have non-positive real part. Since $\Lambda$ is semi-simple, i.e. diagonalizable,
we deduce that
there exists a $\beta>0$ such that for any $\kappa\geq 0$ and any $\phi \in V$,
\begin{equation}\label{eq:21ott03}
| e^{ \kappa \tilde{\Lambda} } \phi | \leq \beta |\phi|
\,.
\end{equation}
Since $|\widetilde{A}(x)|=O(x^{-1-M})$,
we have that
$|\widetilde{A}| \in L^{1}(\mathit{c}, |\dot{\mathit{c}}(t)|dt)$
for every admissible path $\mathit{c}$ , and the function
\begin{equation}\label{eq:21ott03_bis}
\rho_{\mathit{c}}(x)=\beta \int_x^{ \infty} |\widetilde{A}(\mathit{c}(s))||\dot{\mathit{c}}(s)| ds
\end{equation}
is well-defined and satisfies $\lim_{x \to \infty} \rho_\mathit{c}(x)=0$.
Moreover, the inequality \eqref{eq:21ott03} implies that
\begin{equation*}
\left|\int_x^{\infty}e^{-\tilde{\Lambda}\big( S(x,E)-S(y,E) \big)} \widetilde{A}(y) U(y) dy\right|
 \leq \rho_{\mathit{c}}(x)\|U\|_{\infty} ,
\end{equation*}
for every $U\in\mc B$.
By \eqref{eq:21ott03} and \eqref{eq:21ott03_bis}
it follows \cite{fedoryuk93} that $K[U]$ is path-independent and thus well defined.
Furthermore, we have
$$
|K[U](x)| \leq \|U\|_{\infty} \rho(x)
\quad\text{ and }  \|K\|=\sup_{|U|_{\infty}=1}\|K[U]\|_{\infty}\leq \bar\rho
\,,
$$
where
$$
\rho(x)=\inf_{\mathit{c}\text{\scriptsize{ admissible}}}\rho_{\mathit{c}}(x)\,,
\qquad
\bar\rho=\sup_{x \in D_{\e}}\rho(x)
\,.
$$
Note that $\rho(x)$ is a bounded function such that $\lim_{x \to \infty} \rho(x)=0$.
Hence, we conclude that $K$ is a well-defined bounded operator on $\mc{B}$. By definition,
$K^n[U]$ can be written in terms of $K(x,y)$ as
\begin{equation*}
K^n[U](x)
=\int_x^{\infty} \dots \int_{y_{n-1}}^{\infty} K(x,y_1)\dots K(y_{n-1},y_n) U(y_n) dy_1 \dots dy_n
\,,
\end{equation*}
where $K(x,y)$ is the matrix-valued function
$$
K(x,y)=e^{-\tilde{\Lambda} \big( S(x,E)-S(y,E) \big)} \widetilde{A}(y)
\,,
$$
and from this it follows that $K^n[U](x)$ can be estimated as
\begin{equation*}
 |K^n[U](x)|
=\|U\|_{\infty}  \int_x^{\infty} \!\!\!\!\dots \!\int_{y_{n-1}}^{\infty}
|\widetilde{A}(y_1)|\dots |\widetilde{A}(y_n)|  dy_1 \dots dy_n
\leq \|U\|_{\infty} \frac{\rho^n(x)}{n!}
\,.
\end{equation*}
Thus
\begin{equation}
 \label{eq:22ott03}
\|K^n\|\leq \frac{\bar{\rho}^{\,n}}{n!} \, ,
\end{equation}
and the series
$$
\Phi= \sum_{n\in\mb Z_+} K^n[\psi]
$$
converges in $\mc{B}$. Clearly $\Phi$ satisfies the integral equation \eqref{eq:21ott05},
since
$$
K[ \Phi]= \sum_{n=1}^{\infty} K^n[\psi]= \Phi - \psi
\,.
$$
Moreover, $\Phi$ is the unique solution of equation \eqref{eq:21ott05}, for if $\tilde{\Phi}$ is another solution then
$$K^n[\Phi-\tilde{\Phi}]=\Phi - \tilde{\Phi} \,,
\qquad\text{for every }n \in \bb{Z}_+
\,.
$$
By equation \eqref{eq:22ott03}, it follows that $\Phi=\tilde{\Phi}$.
\\

\noindent
It remains to prove that $\Phi$ is an entire function of $E$.
This follows, by a standard perturbation theory argument,
from the fact that $\widetilde{A}$ and $S$ are holomorphic functions of $E$ in $D_{\e}$. Theorem \ref{thm:asymptotic} is proved.

\section{The \texorpdfstring{$\Psi$}{Psi}-system}\label{sec:psi-sistem}
In this section, for a simple Lie algebra $\mf g$ of $ADE$ type, we prove an algebraic identity known as 
$\Psi$-system, which was conjectured first in \cite{dorey07}. 
Recall from Section \ref{sec:kac_moody} that we can write
\begin{equation}\label{eq:pisa_direct_sum}
\bigwedge^2V^{(i)}_{\frac{1}{2}} \cong W^{(i)}\oplus U\,,
\qquad
M^{(i)}=\bigotimes_{j\in I}\big(V^{(j)}\big)^{\otimes B_{ij}} \cong  W^{(i)}\oplus\widetilde U\,,
\end{equation}
where the representation $W^{(i)}=L(\eta_i)_{\frac{p(i)+1}{2}}$, $i\in I$, was defined in equation \eqref{eq:evWi}, and
$U$ and $\widetilde U$ are direct sums of all the irreducible representations
different from $W^{(i)}$. Due to \eqref{m_i-affine}, for every $i\in I$ we have a morphism of representations 
$$
m_i:\bigwedge V_{\frac12}^{(i)}\to M^{(i)}
\,, \quad \ker m_i=U \,.
$$
The $\Psi$-system is a set of quadratic relations, realized by means of the morphisms $m_i$, among the subdominant
solutions to the linear ODE
\eqref{20141125:eq1} defined on the representations $\bigwedge V_{\frac12}^{(i)}$ and $M^{(i)}$.

In order to prove the main result of this section, we
need to study the maximal eigenvalue (and the corresponding
eigenvector) of $\Lambda$ in the representations
$\bigwedge^2V^{(i)}_{\frac{1}{2}}$, $M^{(i)}$
and $W^{(i)}$. To this aim, we recall some further facts from the theory of simple Lie algebras and
simple Lie groups. First, we need the following
\begin{lemma}\label{themostimportantlemma}
Let $h\in \mf h$ satisfy relations \eqref{20141020:eq2}, and let us set 
$\gamma=e^{\frac{2\pi i}{h^\vee}}$.
Then, for every $k\in\mb C$, the following formula holds:
\begin{equation}\label{20150108:eq3}
\gamma^{k\ad h}\Lambda=\gamma^k\mc M_{-k}(\Lambda)
\,.
\end{equation}
In particular, if $\psi$ is an eigenvector of $\Lambda$ in an evaluation representation $V$, then
$\psi_k=\gamma^{-k h}\psi$ is an eigenvector of $\Lambda$ in the representation $V_k$, and we have
\begin{equation}\label{20150108:eq6}
\Lambda \psi=\lambda \psi \qquad
\text{if and only if}\qquad  \Lambda \psi_k = \gamma^k \lambda \psi_k \, .
\end{equation}
\end{lemma}
\begin{proof}
It follows directly from the commutation relations \eqref{20141020:eq2}.
\end{proof}
Now consider the element $\bar\Lambda=\Lambda|_{t=1}\in\mf g$, which is well known \cite{Kos59}
to be a regular semisimple element of $\mf g$. Therefore, its
centralizer $\mf g^{\bar\Lambda}=\{a\in\mf g\mid [\bar\Lambda,a]=0\}$
is a Cartan subalgebra of $\mf g$, which we denote $\mf h'$.
Due to Lemma \ref{themostimportantlemma}, it follows that
\begin{equation}\label{eq:lambdacyclic}
\gamma^{\ad h} \bar \Lambda=\gamma \bar\Lambda\,,
\end{equation}
where $h\in\mf h$ is the element introduced in Section \ref{sec:1} 
satisfying relations \eqref{20141020:eq2}. Equation \eqref{eq:lambdacyclic} says that $\bar \Lambda$ is an eigenvector
of $\gamma^{\ad h}$, and this in turn
 implies that  $\mf{h}'$ is stable under the action of $\gamma^{\ad h}$. Since  $\gamma^{\ad h}$ is a Coxeter-Killing
 automorphism of $\mf g$ \cite{Kos59},
  we can regard $\gamma^{\ad h}$ as a Coxeter-Killing automorphism of $\mf h'$.
From these considerations it follows -- as proved in \cite{moody87} -- that
we can choose a set of Chevalley generators of $\mf g$, say $\{f_i',h_i',e_i'\mid i\in I\}\subset\mf g$,
such that $\mf h'=\bigoplus_{i\in I}\mb Ch_i'$ and that the identity 
\begin{equation}\label{eq:coxeterincidence}
2\mbbm1_n+\gamma^{\ad h}+\gamma^{-\ad h}=(B^t)^2 \, ,
\end{equation}
holds as operators on $\mf h'$. Here $B$ denotes, as usual, the incidence matrix of $\mf g$. Following \cite{olive92},
relation \eqref{eq:coxeterincidence}
can be used to express -- in the basis $\left\{h'_i, i\in I\right\}$ -- the eigenvectors of $\gamma^{\ad h}$  in terms of the
eigenvectors of $B$. Now it is well known that the incidence matrix $B$
is a non-negative irreducible matrix  \cite{Kac90}. Therefore, it has a maximal (in the sense of the Perron-Frobenius theory)
eigenvalue, which is 
$\gamma^{\frac{1}{2}}+\gamma^{-\frac{1}{2}}$,
and a corresponding eigenvector, say $(x_1,x_2,\dots,x_n)\in\mb R_{>0}^n$.
If we fix $x_1=1$, we then have the following relation for every $i\in I$:
\begin{equation}\label{eq:eigenB}
\sum_{j \in I} B_{ij} x_j
= (\gamma^{-\frac{1}{2}}+\gamma^{\frac{1}{2}}) x_i
\,,
\qquad x_i>0
\,.
\end{equation}
Following \cite{olive92}, and essentially using \eqref{eq:coxeterincidence},
we finally obtain
\begin{equation}\label{eq:decompositionLambda}
\bar{\Lambda}=\sum_{i \in I} \gamma^{\frac{p(i)}{2}}x_i h'_i\,,
\end{equation}
where $x_i, i \in I$, are defined in \eqref{eq:eigenB} and the function $p(i)$, $i\in I$, is given in  \eqref{eq:p}.
\begin{remark}\label{rem:quandofiniremodiscrivere?}
Note that equation \eqref{eq:decompositionLambda} corresponds to a precise normalization of
$\bar\Lambda$.
In fact, it is always possible to find an automorphism of $\mf g$, under which $\mf h$ is stable,
in such a way that
$\bar{\Lambda}=c \sum_{i \in I} \gamma^{\pm \frac{p(i)}{2}}x_i h'_i$ for any $c \in \bb{C}^*$ and any choice of the
relative phase
$\gamma^{\pm \frac{p(i)}{2}} $.
The choice of the relative phase is however immaterial. Indeed, under the transformation $\gamma^{\frac{p(i)}{2}}
\to \gamma^{-\frac{p(i)}{2}} $ the right hand side of \eqref{eq:decompositionLambda} represents a different element
in the Cartan subalgebra $\mf h'$ which is however conjugated to $\bar{\Lambda}$. Hence, their spectra coincide. 
\end{remark}
The decomposition \eqref{eq:decompositionLambda} allows one to compute the spectrum of $\Lambda$ using the Perron-Frobenius eigenvector of the incidence matrix $B$. We will also need the following:
\begin{lemma}\label{lem:davveroultimo}
Let $\alpha_j$, $j \in I$, be the positive roots of $\mf g$ with respect to the Cartan
subalgebra $\mf{h}'$. Then $\bar{\Lambda}$ satisfies
\begin{equation}\label{eq:alphajlambda}
\alpha_j(\gamma^{-\frac{p(i)}{2}}\bar{\Lambda})= x_j \gamma^{\frac{p(j)-p(i)}{2}}(1-\gamma^{-2p(j)+1}) \,, \quad
\text{for all } i,j \in I
\,,
\end{equation}
where the $x_i$ are defined in \eqref{eq:eigenB}. This implies that
\begin{align*}
& \Re \alpha_j(\gamma^{-\frac{p(i)}{2}}\bar{\Lambda}) >0\,, \quad \mbox{ if }\quad p(i)=p(j)\,, \\ & \Re \alpha_j(\gamma^{-\frac{p(i)}{2}}\bar{\Lambda}) =0\,, \quad \mbox{ if } \quad p(i) \neq p(j)\,.
\end{align*}
\begin{proof}
 Since $B_{ij}\neq0$ if and only if $p(i)\neq p(j)$, equation \eqref{eq:eigenB} implies that
\begin{equation}\label{20150124:eq1}
\sum_{j\in I}\gamma^{-\frac{p(j)}{2}}B_{ij}x_j
=(\gamma^{-\frac12}+\gamma^{\frac12})\gamma^{\frac{p(i)-1}{2}}x_i
\,.
\end{equation}
Equation \eqref{20150124:eq1} and the decomposition \eqref{eq:decompositionLambda} imply the thesis.
\end{proof}
\end{lemma}
Using the above lemma, together with equation \eqref{eq:decompositionLambda}, we compute the maximal eigenvalue
of $\Lambda$ in the representations $V^{(i)}$, $\bigwedge V^{(i)}$, $W^{(i)}$ and $M^{(i)}$.
\begin{proposition}\label{prop:pisa2}
Let $i\in I$. If $\mf g$ is a simple Lie algebra of $ADE$ type, then the following facts hold.
\begin{enumerate}[(a)]
\item
For the fundamental representation $V^{(i)}$
there exists a unique maximal eigenvalue $\lambda^{(i)}$ of $\Lambda$.
The eigenvalues $\lambda^{(i)}, i\in I$, satisfy the following identity:
\begin{equation}\label{eq:Blambda}
\sum_{j\in I}B_{ij}\lambda^{(j)}=(\gamma^{-\frac12}+\gamma^\frac12)\lambda^{(i)}.
\end{equation}
We denote by $\psi^{(i)}\in V^{(i)}$ the unique (up to a constant factor)
eigenvector of $\Lambda$ corresponding to the eigenvalue $\lambda^{(i)}$ . 
\item
For the representation $M^{(i)}$,
we have that $\sum_{j\in I}B_{ij}\lambda^{(i)}$
is a maximal eigenvalue of $\Lambda$.
Moreover,
$$
\psi^{(i)}_\otimes=\bigotimes_{j\in I}{\psi^{(j)}}^{\otimes B_{ij}}\in M^{(i)}
$$
is the corresponding eigenvector.
\item
For the representation $\bigwedge^2 V_{\frac12}^{(i)}$,
we have that $(\gamma^{-\frac12}+\gamma^\frac12)\lambda^{(i)}$
is a maximal eigenvalue of $\Lambda$.
Moreover,
$$
\gamma^{-\frac{1}{2} h } \psi^{(i)} \wedge
\gamma^{\frac{1}{2} h }
\psi^{(i)} \in\bigwedge^2V^{(i)}_{\frac{1}{2}} 
$$
is the corresponding eigenvector.
\item
The elements $\gamma^{-\frac{1}{2} h } \psi^{(i)} \wedge
\gamma^{\frac{1}{2} h }
\psi^{(i)} \in\bigwedge^2V^{(i)}_{\frac{1}{2}} $ and
$\psi^{(i)}_\otimes\in M^{(i)}$
belong to the subrepresentation $W^{(i)}$.
Hence, for the representation $W^{(i)}$
there exists a maximal eigenvalue $\mu^{(i)}$ of $\Lambda$,
given by
$$
\mu^{(i)}=\sum_{j\in I}B_{ij}\lambda^{(i)}=(\gamma^{-\frac12}+\gamma^\frac12)\lambda^{(i)}
\,,
$$
and the relation
$$
m_i\big( \gamma^{-\frac{1}{2} h } \psi^{(i)} \wedge\gamma^{\frac{1}{2} h } \psi^{(i)} \big)
=c \psi^{(i)}_\otimes
$$
holds for some $c\in\mb C^*$.
\end{enumerate}
\end{proposition}
\begin{proof}
Recall that the weights appearing in the representation $L(\omega_i)$ are of the following type:
$$\omega_i, \qquad \omega_i-\alpha_i, \qquad \omega_i-\alpha_i-\alpha_j -\omega\,,$$ where in the latter case the index $j$ is such that
$B_{ij}\neq0$ and $\omega$ is an integral non-negative linear combination of the positive roots.
Clearly for any $\tilde{h} \in \mf{h}'$ such that $\Re \alpha_j(\tilde{h})\geq0$, $\forall j\in I$ we have
$$ 
\Re \omega_i(\tilde{h})\geq  (\omega_i-\alpha_i)(\tilde{h}) \geq (\omega_i-\alpha_i-\alpha_j -\omega)(\tilde{h}) \, .
$$
Therefore for any such $\tilde{h} $ the eigenvalues with maximal real part is $\omega_i(\tilde{h})$, followed by
$(\omega_i-\alpha_i)(\tilde{h})$, $(\omega_i-\alpha_i-\alpha_j)(\tilde{h})$, and so on.

We use this simple argument to study the maximal eigenvalues of $\Lambda$ in the representation $V^{(i)}$.
By definition of $\Lambda$ and $\bar \Lambda$, and using equation \eqref{20150108:eq3},
the action of $\Lambda$ on $V^{(i)}$ coincides with the action of $\gamma^{\frac{p(i)}{2}(-1+\ad h)} \bar{\Lambda}$
on $L(\omega_i)$.
By Lemma \ref{lem:davveroultimo}, $\alpha_j(\gamma^{-\frac{p(i)}{2}} \bar{\Lambda})\geq 0$,
for every $j \in I$,
and $\alpha_i(\gamma^{-\frac{p(i)}{2}} \bar{\Lambda})=x_i(1-\gamma^{1-2p(i)})$.
It follows that in the representation $V^{(i)}$, $\Lambda$ has the maximal eigenvalue
$$\omega_i(\gamma^{-\frac{p(i)}{2}} \bar{\Lambda})=x_i \, .$$
Therefore, by definition of the $x_i$'s, the $\lambda^{(i)}$ satisfy equation \eqref{eq:eigenB}.

Moreover we have 
\begin{equation}\label{eq:psiivi}
\psi^{(i)}=\gamma^{\frac{p(i)}2 h} v'_i
\end{equation}
and 
\begin{equation}\label{eq:fpsiivi}
f'_i\psi^{(i)}= c \gamma^{-h}\psi^{(i)}\,,
\end{equation}
for some $c \in \bb C^*$. This completes the proof of part (a).
Part (b) follows directly from part (a).
To prove part (c), we follow the same lines we used to prove part (a). In fact,
the action of $\Lambda$ on $V^{(i)}_{\frac{1}{2}}$ coincides with the action
of $\gamma^{\frac{p^*(i)}{2}(-1+\ad h)} \bar{\Lambda}$
on $L(\omega_i)$, where $p^*(i)=1-p(i)$.
By equation \eqref{eq:alphajlambda} we have $\Re \alpha_i(\gamma^{\frac{p^*(i)}{2}(-1+\ad h)} \bar{\Lambda})=0$, while
$\Re \alpha_j(\gamma^{\frac{p^*(i)}{2}(-1+\ad h)} \bar{\Lambda})>0 $ if $p(i) \neq p(j)$. It follows that
the two eigenvalues with maximal real part correspond to
to the weights $\omega_i$ and $\omega_i-\alpha_i$.  By equation \eqref{eq:alphajlambda}, these eigenvalues are $x_i \gamma^{\frac{1}{2}}$ and 
$x_i \gamma^{-\frac{1}{2}}$, and the corresponding eigenvectors are $ \gamma^{\mp \frac12h}\psi^{(i)}$.
Therefore $\bigwedge V^{(i)}$ has a unique maximal eigenvalue $ (\gamma^{\frac12}+\gamma^{-\frac12})x_i $, with
eigenvector
\begin{equation}\label{eq:psiunmezzo}
\gamma^{-\frac12h}\psi^{(i)}\wedge\gamma^{\frac12h}\psi^{(i)}=c \gamma^{\frac{(1+p(i))}2h}(v'_i \wedge f'_iv'_i)\,,
\end{equation}
for some $c \neq 0$. Here we have used equations \eqref{eq:psiivi} and  \eqref{eq:fpsiivi}.

Let us prove (d). By equation \eqref{eq:psiivi} and Lemma \ref{lem:231202},
it follows that $\psi^{(i)}_\otimes \in W^{(i)}$, where $W^{(i)}$ is seen as a subrepresentation
of $M^{(i)}$. Therefore the maximal eigenvalue of $W^{(i)}$ and $M^{(i)}$ coincide.
Finally by equations \eqref{eq:psiivi}, \eqref{eq:psiunmezzo} and Lemma \ref{lem:021202}, 
we have that $\gamma^{-\frac12h}\psi^{(i)}\wedge\gamma^{\frac12h}\psi^{(i)} \in W^{(i)}$
and $m_i(\gamma^{-\frac12h}\psi^{(i)}\wedge\gamma^{\frac12h}\psi^{(i)})
=c\psi^{(i)}_\otimes$, $c\in\mb C^*$, thus proving part (d) and concluding the proof.
\end{proof}
\begin{remark}\label{rem:pisamerda1}
We note that we can always choose a normalization of the eigenvectors $\psi^{(i)}\in V^{(i)}$,
$i\in I$, such that
\begin{equation}\label{eq:021202}
m_i\big( \gamma^{-\frac{1}{2} h } \psi^{(i)} \wedge\gamma^{\frac{1}{2} h } \psi^{(i)} \big)
= \psi^{(i)}_\otimes
\,,
\end{equation}
for every $i\in I$. From now on, we will assume that Proposition \ref{prop:pisa2}(d)
holds with the normalization constant $c=1$.
\end{remark}
By Proposition \ref{prop:pisa2}(a), for any fundamental representation $V^{(i)}$, $i\in I$ of
$\widehat{\mf g}$, there exists a maximal eigenvalue $\lambda^{(i)}$ with eigenvector $\psi^{(i)}$.
Therefore, by Theorem \ref{thm:asymptotic}
there exists a unique subdominant solution $\Psi^{(i)}(x,E):\widetilde{\mb C}\to V^{(i)}$
of the differential equation
\eqref{20141125:eq1} in the representation $V^{(i)}$
with asymptotic behavior
\begin{equation}\label{eq:21ott04}
 \Psi^{(i)}(x,E)
 =e^{-\lambda^{(i)} S(x,E)} q(x,E)^{-h}\big( \psi^{(i)} + o(1) \big)
 \,,
\text{ in the sector } |\arg{x}| <\frac{\pi}{2(M+1)}
\,.
\end{equation}
Using Proposition \ref{prop:pisa2} we can prove the following
theorem establishing the $\Psi$-system associated to the Lie algebra $\mf g$.
\begin{theorem}\label{thm:psi-sistem}
Let $\mf g$ be a simple Lie algebra of $ADE$ type, and let the solutions $\Psi^{(i)}(x,E):\widetilde{\mb C}\to V^{(i)}$,
$i\in I$, have the asymptotic behavior \eqref{eq:21ott04}.
Then, the following identity holds:
\begin{equation}\label{eq:031201}
m_i\big(  \Psi_{-\frac{1}{2}}^{(i)}(x,E) \wedge
\Psi_{\frac{1}{2}}^{(i)}(x,E) \big) =\otimes_{j\in I} \Psi^{(j)}(x,E)^{\otimes B_{ij}}\,,
\quad\text{for every }i\in I
\,.
\end{equation}
\end{theorem}
\begin{proof}
Due to Proposition \ref{prop:pisa2}(b) and Theorem \ref{thm:asymptotic},
the unique subdominant solution
to equation \eqref{20141125:eq1} in $M^{(i)}$
is
$$
\otimes_{j\in I} \Psi^{(j)}(x,E)^{\otimes B_{ij}}
=e^{- \mu^{(i)}S(x,E)} q(x,E)^{-h}\left( \psi_\otimes^{(i)}+o(1)\right)\,,
\quad\text{for }x\gg0
\,.
$$
Moreover, by equation \eqref{eq:Psik},
Corollary \ref{cor:asymptotic} and Proposition \ref{prop:pisa2}(c) we have that
$$
\Psi_{-\frac{1}{2}}^{(i)}(x,E) \wedge\Psi_{\frac{1}{2}}^{(i)}(x,E)
=e^{- \mu^{(i)}S(x,E)} q(x,E)^{-h}  \left( \psi_{-\frac12}^{(i)}\wedge\psi_{\frac12}^{(i)} +o(1)\right)\,,
\quad\text{for }x\gg0
\,.
$$
The proof follows by equation \eqref{eq:021202} and the uniqueness of the subdominant solution.
\end{proof}
\begin{example}\label{exa:psi_system_A}
In type $A_n$, equation \eqref{eq:031201} becomes
$$
m_i\big(  \Psi_{-\frac{1}{2}}^{(i)} \wedge
\Psi_{\frac{1}{2}}^{(i)} \big) =\Psi^{(i-1)}\otimes\Psi^{(i+1)}\,,
\quad\text{for }i=1,\dots,n
\,,
$$
where we set $\Psi^{(0)}=\Psi^{(n+1)}=1$.
\end{example}
\begin{example}\label{exa:psi_system_D}
In type $D_n$, equation \eqref{eq:031201} becomes
\begin{align}\label{eq:psidn1}
&m_i\big(  \Psi_{-\frac{1}{2}}^{(i)} \wedge
\Psi_{\frac{1}{2}}^{(i)} \big) =\Psi^{(i-1)}\otimes\Psi^{(i+1)}\,,
\quad\text{for }i=1,\dots,n-3
\,,
\\ \label{eq:psidn2}
&m_{n-2}\big(  \Psi_{-\frac{1}{2}}^{(n-2)} \wedge\Psi_{\frac{1}{2}}^{(n-2)} \big)
=\Psi^{(n-3)}\otimes\Psi^{(n-1)}\otimes\Psi^{(n)}\,,
\\ \label{eq:psidn3}
&m_{n-1}\big(  \Psi_{-\frac{1}{2}}^{(n-1)} \wedge\Psi_{\frac{1}{2}}^{(n-1)} \big)
=\Psi^{(n-2)}
=m_{n}\big(  \Psi_{-\frac{1}{2}}^{(n)} \wedge\Psi_{\frac{1}{2}}^{(n)} \big)
\,,
\end{align}
where we set $\Psi^{(0)}=1$.
\end{example}
\begin{remark}\label{re:jinjun}
The $\Psi$-system in Example \ref{exa:psi_system_D} is a complete set of relations among the
subdominant solutions. The analogue system of equations 
obtained in \cite[Eqs. (3.19,3.20,3.21)]{Sun12} lacks of the equation \eqref{eq:psidn2} for the
node $n-2$ of the Dynkin diagram.  Note that equations  \cite[Eqs. (3.19,3.20)]{Sun12} are  equivalent
to \eqref{eq:psidn1} and \eqref{eq:psidn3} of the present paper.
Indeed, it follows from the direct sum decomposition \ref{eq:pisa_direct_sum} that 
$m_i\circ\iota=\mbbm1_{V^{(n-2)}}$,
for $i=n-1,n$, where $\iota$ is the embedding
$V^{(n-2)}\stackrel{\iota}{\hookrightarrow}\bigwedge^2 V^{(i)}_{\frac12}$
(in fact, $V^{(n-2)}\cong W^{(n-1)}\cong W^{(n)}$ in this particular case).
The further identity \cite[Eq. (3.21)]{Sun12}, which 
can be obtained from \eqref{eq:ultimogiorno}, requires the introduction of another unknown and therefore equations
\cite[Eqs. (3.19,3.20,3.21)]{Sun12} are an incomplete system of $n$ equations in $n+1$ unknowns.
\end{remark}
\begin{example}\label{exa:psi_system_E}
In type $E_n$, $n=6,7,8$, equation \eqref{eq:031201} becomes
\begin{align*}
&m_i\big(  \Psi_{-\frac{1}{2}}^{(i)} \wedge
\Psi_{\frac{1}{2}}^{(i)} \big) =\Psi^{(i-1)}\otimes\Psi^{(i+1)}\,,
\quad\text{for }i=1,\dots, n-4
\,,
\\
&m_{n-3}\big(  \Psi_{-\frac{1}{2}}^{(n-3)} \wedge\Psi_{\frac{1}{2}}^{(n-3)} \big)
=\Psi^{(n-4)}\otimes\Psi^{(n-2)}\otimes\Psi^{(n-1)}\,,
\\
&m_{n-2}\big(  \Psi_{-\frac{1}{2}}^{(n-2)} \wedge
\Psi_{\frac{1}{2}}^{(n-2)} \big) =\Psi^{(n-3)}
\,,
\\
&m_{n-1}\big(  \Psi_{-\frac{1}{2}}^{(n-1)} \wedge
\Psi_{\frac{1}{2}}^{(n-1)} \big) =\Psi^{(n-3)}\otimes\Psi^{(n)}
\,,
\\
&m_n\big(  \Psi_{-\frac{1}{2}}^{(n)} \wedge
\Psi_{\frac{1}{2}}^{(n)} \big) =\Psi^{(n-1)}
\,,
\end{align*}
where we set $\Psi^{(0)}=1$.
\end{example}

\section{The \texorpdfstring{$Q$}{Q}-system}\label{sec:Q}
In this section, for any simple Lie algebra of $ADE$ type,
we define the generalized spectral determinants $Q^{(i)}(E;\ell)$, $i\in I$,
and we prove that they satisfy the Bethe Ansatz equations, also known as
$Q$-system. The spectral determinants are entire functions of the parameter $E$, and
they are defined by the behavior of the functions $\Psi^{(i)}(x,E)$ close to the Fuchsian singularity.
More precisely, $Q^{(i)}(E;\ell)$ is the coefficient of the most singular term of the expansion of $\Psi^{(i)}(x,E)$
in the neighborhood of $x=0$.

In the case of the Lie algebra $\mf{sl}_2$,  the spectral determinant $Q(E;\ell)$ was originally introduced
in \cite{dorey98,bazhanov01}, while the generalization to $\mf{sl}_n$ has been given in \cite{dorey00,junji00}.
An incomplete construction for Lie algebras of $BCD$ type can be found in \cite{dorey07,Sun12}. The terminology generalized
spectral determinants is motivated by the $\mf{sl}_2$ case.
Indeed, for the Lie algebra $\mf{sl}_2$, the linear ODE \eqref{20141125:eq1}
is equivalent to a Schr\"odinger equation with a polynomial potential and a
centrifugal term, and the spectral determinant $Q(E;\ell)$ vanishes
at the eigenvalues of the Schr\"odinger operator.

Before proving the main result of this section, we briefly review
some well-known fundamental results from the theory of Fuchsian singularities. Here we follow \cite[Chap. III]{Ilya08}.

\begin{remark}\label{rem:integerpotential}
 In this Section we assume that the potential $p(x,E)=x^{M h^\vee}$ is analytic in $x=0$,
 namely $Mh^\vee \in \bb{Z}_+$. With this assumption,
 the point $x=0$ is simply a Fuchsian singularity of the equation. The case of a potential with a branch point
 at $x=0$ can formally be treated without any modification -- see \cite{bazhanov01} for the case $A_1$ -- but the mathematical theory
 is less developed.
\end{remark}

\subsection{Monodromy about the Fuchsian singularity}\label{sec:monodromy}
For any linear ODE with Fuchsian singularity at $x=0$, namely of the type
\begin{equation}\label{eq:fuchsian}
\Psi'(x)=\left(\frac Ax +\ B(x)\right) \Psi(x)
\,,
\end{equation}
where $A$ is a constant matrix and $B(x)$ is a regular function,
the monodromy operator is the endomorphism on the space of solutions of
\eqref{eq:fuchsian}
which associates to any solution its analytic continuation
around a small Jordan curve encircling $x=0$. More concretely, if we fix a matrix solution $Y(x)$ of the linear ODE \eqref{eq:fuchsian}
with non-vanishing determinant (i.e. a basis of solutions), then
the monodromy matrix $M$ is defined by the relation
\begin{equation*}
 Y(e^{2 \pi i}x)=Y(x)M \, .
\end{equation*}
Here, $Y(e^{2 \pi i}x)$ is the customary notation for the analytic continuation of $Y(x)$ . We are interested in the
invariant subspaces of the monodromy matrix $M$.
By the general theory, the algebraic eigenvalues of the monodromy matrix $M$
are of the form $e^{2\pi ia}$, for any eigenvalue $a$ of $A$.

We said that an eigenvalue $a$ of $A$ is non-resonant if there exists no other eigenvalue of $a'$ such that
$a-a' \in \bb{Z}$.
In this case the eigenvalue $e^{2\pi i a}$ has multiplicity one and there exists a solution of
equation \eqref{eq:fuchsian} which is an eigenvector of the monodromy matrix $M$
with eigenvalue $e^{2\pi i a}$. Such a solution has the form
\begin{equation}\label{eq:simplefuchsian}
 \chi_{a}(x)=x^{a}\big( \chi_{a} +O(x) \big)
\end{equation}
for any $\chi_a$ eigenvector of $A$ with eigenvalue $a$.
If the eigenvalues $a_1,\dots, a_k$ are resonant,
meaning that $a_i -a_j \in \bb{Z}$ for $i,j=1\dots,k$, then
the eigenvalue $e^{2 \pi i a_1}=\dots=e^{2\pi i a_k}$ has multiplicity $k$. 
The matrix $M$, in general, is not diagonalizable when restricted to the $k$-dimensional invariant
subspace corresponding to this eigenvalue.
The corresponding solutions
do not have the form \eqref{eq:simplefuchsian}
due to the appearance of logarithmic terms.

We are interested in the linear ODE \eqref{20141125:eq1}, which has a Fuchsian singularity with $A=-\ell $, $\ell\in\mf h$.
For any representation $V^{(i)}$, $i\in I$,
the eigenvalues for the action of $\ell$ on $V^{(i)}$ have the form $\lambda(\ell)$,
where $\lambda\in P$ is a weight of the representation $V^{(i)}$.  Note that for a generic choice of
$\ell\in\mf h$ and  for distinct weights $\lambda_1$ and $\lambda_2$,
the corresponding eigenvalues $\lambda_1(\ell)$, $\lambda_2(\ell)$ are non-resonant.
Indeed, generically, two eigenvalues are resonant if and only if they coincide,
and therefore they correspond to a weight of multiplicity bigger than one.

\subsection{The dominant term at \texorpdfstring{$x=0$}{x=0}}\label{sec:dominant}
Let $\omega\in P^+$,
and
let us denote by
$$
l^*=\max_{\lambda\in P_\omega} \Re\lambda(\ell)\,.
$$
Let also $\lambda^*\in P$ be such that $l^*=\Re\lambda^*(\ell)$.
If $\lambda^*$ is a weight of multiplicity one in $L(\omega)$, then
-- by the discussion in Section \ref{sec:monodromy} --
the most singular behavior at $x=0$
of a solution to equation \eqref{20141125:eq1} is $x^{-l^*}$. In this section, for any irreducible
representation $L(\omega)$ of $\mf{g}$ and a generic $\ell \in \mf{h}$,  we characterize the weight $\lambda^*\in P$
maximizing $\lambda(\ell)\in\mb Z$.
Moreover, we show that it has multiplicity one.
These facts will be used to define the generalized spectral determinant $Q^{(i)}(E;\ell)$.

It is well-known from the general theory of simple Lie algebras (see \cite[Appendix D]{FH91}),
that
given a generic element $\ell \in \mf h$,
we can decompose the set of the roots $R$ of $\mf g$
in two distinct and complementary
sets (of positive and negative roots)
$$
R^+_\ell=\lbrace \alpha\in R\mid \Re \alpha (\ell)>0 \rbrace\,,
\qquad
R^-_\ell= \lbrace \alpha\in R\mid\Re\alpha(l)<0 \rbrace
\,.
$$
Consequently, we can associate to this $\ell\in\mf h$ a Weyl Chamber $\mf{W}_\ell$, as well as  an element
$w_\ell$ of the Weyl group (it is the element which maps
the original Weyl chamber into $\mf{W}_\ell$) and a (possibly new) set of simple roots
$\Delta_\ell=\lbrace w_\ell(\alpha_i)\mid i\in I \rbrace$. The weight $w_\ell(\omega)$ is the highest weight
of $L(\omega)$ with respect to
the new set of simple roots $\Delta_\ell$, and due to the definition of $R^-_\ell$, the action of any negative
simple root $w_\ell(-\alpha_i)$, $i\in I$,
decreases the value of $\Re w_\ell(\omega)(\ell)$. 
Therefore, $\lambda^*=w_\ell(\omega)$ is the unique weight maximizing the
function $\Re \lambda(\ell)$, $\lambda\in P_\omega$.
This weight has multiplicity one since it lies
in the Weyl orbit of $\omega$.
\\

In the case of a fundamental representation $L(\omega_i)$, $i\in I$,
the weight $w_\ell(\omega_i-\alpha_i)$ belongs to $P_{\omega_i}$,
and it has multiplicity one.
Moreover, it is possible to show that all the weights in $P_{\omega_i}$ different
from $w_\ell(\omega_i)$ are obtained from $w_\ell(\omega_i-\alpha_i)$ by a repeated action of the
negative simple roots $w_\ell(-\alpha_i)$.
We conclude that $w_\ell(\omega_i-\alpha_i)$ is the unique weight maximizing
$\Re \lambda(\ell)$ on $P_{\omega_i}\setminus \lbrace w_\ell(\omega_i) \rbrace$. We have thus proved the following result:
\begin{proposition}\label{prop:maximalweight}
Let $\ell \in \mf h$ be a generic element, and let $w_\ell$ the associated element of the Weyl group.
Then the weight $w_\ell(\omega)\in P_\omega$, $\omega\in P^+$, has multiplicity one and
$$
\Re w_\ell(\lambda)(\ell) > \Re \lambda(\ell)
\,, 
$$
for any weight $\lambda \in P_{\omega}$, $\lambda\neq w_\ell(\omega)$.
In the case of a fundamental weight $\omega=\omega_i$, $i\in I$,
the weight $ w_\ell(\omega_i-\alpha_i)\in P_{\omega_i}$
has multiplicity one and
$$
\Re w_\ell(\omega_i)(\ell) > \Re w_\ell(\omega_i-\alpha_i)(\ell) > \Re \lambda(\ell) 
\,,
$$
for any $\lambda \in P_{\omega_i}$, $\lambda\neq w_\ell(\omega_i),w_\ell(\omega_i-\alpha_i)$.
\end{proposition}
Let  $\chi^{(i)},\phi^{(i)}\in L(\omega_i)$ be weight vectors corresponding to the the weights
$w_\ell(\omega_i)$ and $w_\ell(\omega_i-\alpha_i)$ respectively.
As done in Lemmas \ref{lem:021202} and \ref{lem:231202}
it is possible to show that 
$L(\eta_i)$ is a subrepresentation of $\bigwedge^2L(\omega_i)$ (respectively
$\bigotimes L(\omega_i)^{\otimes B_{ij}}$) and that
$\chi^{(i)}\wedge\phi^{(i)}\in L(\eta_i)$
(respectively $\otimes_{j\in I} {\chi^{(j)}}{ \otimes B_{ij}}\in L(\eta_i)$).
Moreover, $\chi^{(i)}\wedge\phi^{(i)}$ and $\otimes_{j\in I} {\chi^{(j)}}^{ \otimes B_{ij}}$
are highest weight vectors (with respect to $\Delta_\ell$) of the subrepresentation $L(\eta_i)$.
Hence, we can identify these vectors (up to a constant) using the morphism $m_i$
defined in equation \eqref{m_i}.
From now on we fix the normalization of $\chi^{(i)},\phi^{(i)}$, $i \in I$, in such a way that
\begin{equation}\label{eq:chinormalization}
m_i(\chi^{(i)} \wedge \phi^{(i)} )= \otimes_{j\in I} \chi^{(j) \otimes B_{ij}}
\,.
\end{equation}

\subsection{Decomposition of \texorpdfstring{$\Psi^{(i)}$}{Psi^{(i)}}}
From the above discussion, and in particular form Proposition \ref{prop:maximalweight} and equation \eqref{eq:simplefuchsian},
we have that for any evaluation representation $V^{(i)}_k$, $i\in I$ and $k\in\mb C$,
there exist distinguished (and normalized) solutions
$\chi^{(i)}_k(x,E)$ and $\phi^{(i)}_k(x,E)$
of the linear ODE \eqref{20141125:eq1}
such that they have the most singular behavior at $x=0$ and they are eigenvectors
of the associated monodromy matrix.
Explicitly, these solutions have the asymptotic expansion
\begin{align*}
& \chi^{(i)}_{k}(x,E)= x^{-w_\ell(\omega_i)(\ell)}\big( \chi^{(i)} +O(x) \big), \\ 
& \phi^{(i)}_{k}(x,E)= x^{-w_\ell(\omega_i-\alpha_i)(\ell)}\big( \phi^{(i)} +O(x) \big).
\end{align*}
Since the parameter $E$ appears linearly in the ODE \eqref{20141125:eq1}
and the asymptotic behavior at $x=0$ does not depend on $E$,
by the general theory \cite{Ilya08} it follows that $\chi^{(i)}_{k}(x,E;\ell)$ and
$\phi^{(i)}_{k}(x,E;\ell)$ are entire functions of $E$.

Using equation \eqref{20150108:eq8} and comparing the asymptotic behaviors we have that
\begin{align}
\begin{split}
& \omega^{-k h} \chi_{k'}^{(i)} (\omega^k x,\Omega^k E)
= \omega^{-k w_\ell(\omega_i)(\ell+h)} \chi^{(i)}(x,E)_{k+k'} \\ \label{eq:chik}
& \omega^{-k h} \phi_{k'}^{(i)}(\omega^k x,\Omega^k E)
= \omega^{-k w_\ell(\omega_i-\alpha_i)(\ell+h)} \phi^{(i)}(x,E)_{k+k'}
\,,
\end{split}
\end{align}
for every $k,k'\in\mb C$.
For generic $\ell\in\mf h$,
the eigenvalues $-w_\ell(\omega_i)(\ell)$, and $-w_\ell(\omega_i-\alpha_i)(\ell)$
are non-resonant and therefore
we can unambiguously define two functions $Q^{(i)}(E;\ell)$ and $\widetilde{Q}^{(i)}(E;\ell)$
as follows. We write
\begin{equation}\label{psigrande}
\Psi^{(i)}(x,E,\ell)
=Q^{(i)}(E;\ell)\chi^{(i)}(x,E)+\widetilde{Q}^{(i)}(E;\ell) \phi^{(i)}(x,E)+v^{(i)}(x,E)
\,,
\end{equation}
where $v^{(i)}(x,E)$ belongs to an invariant subspace of the monodromy matrix
corresponding to lower weights.
We call $Q^{(i)}(E;\ell)$ and $\widetilde Q^{(i)}(E;\ell)$ the generalized
spectral determinants of the equation \eqref{20141020:eq1}.
Since $\Psi^{(i)}(x,E,\ell)$, $\chi^{(i)}(x,E)$ and $\phi^{(i)}(x,E)$
are entire functions of the parameter $E$, then also $Q^{(i)}(E;\ell)$
and $\widetilde{Q}^{(i)}(E;\ell)$ are entire functions of the parameter $E$.

Finally, the $\Psi$-system \eqref{eq:031201} implies some remarkable functional
relations among the generalized spectral determinants.
Indeed, we have the following result.
\begin{theorem}\label{thm:QQtilde}
Let $\ell\in\mf h$ be a generic element.
Then, the spectral determinants $Q^{(i)}(E;\ell)$ and $\widetilde{Q}^{(i)}(E;\ell)$
are entire functions of $E$ and they satisfy the following 
$Q\widetilde{Q}$-system:
\begin{align}
\begin{split}\label{eq:QQtilde}
\prod_{j \in I}Q^{(j)}(E;\ell)^{B_{ij}}
&=\omega^{-\frac{1}{2} \theta_i}
Q^{(i)}(\Omega^{-\frac{1}{2}}E;\ell)\widetilde{Q}^{(i)}(\Omega^{\frac{1}{2}}E;\ell)
\\
& -\omega^{\frac{1}{2} \theta_i }
Q^{(i)}(\Omega^{\frac{1}{2}}E;\ell)\widetilde{Q}^{(i)}(\Omega^{-\frac{1}{2}}E;\ell)
\,,
\end{split}
\end{align}
where $\theta_i=w_\ell(\alpha_i)(\ell +h)$.
\end{theorem}
\begin{proof}
We already showed that $Q^{(i)}(E;\ell)$
and $\widetilde{Q}^{(i)}(E;\ell)$ are entire functions of the parameter $E$.
By the definition of $\Psi^{(i)}_{\pm \frac{1}{2}}(x,E;\ell)$ given in \eqref{eq:Psik}
and from equations \eqref{psigrande} and  \eqref{eq:chik} we have that
\begin{align}
\begin{split}\label{questoeunlabelconunnomeumanomauncuorenichilista}
\Psi_{\pm \frac{1}{2}}^{(i)}(x,E;\ell)
&= \omega^{\pm \frac{1}{2} w_\ell(\omega_i)(\ell+h)}
Q^{(i)}(\Omega^{\pm \frac{1}{2}}E;\ell)\chi^{(i)}_{ \frac{1}{2}}(x,E)\\
&+ \omega^{\pm \frac{1}{2} w_\ell(\omega_i-\alpha_i)(\ell+h)}
\widetilde{Q}^{(i)}(\Omega^{\pm \frac{1}{2}}E;\ell)\phi^{(i)}_{\frac{1}{2}}(x,E) +
\tilde v^{(i)}(x,E)
\,.
\end{split}
\end{align}
Here
we used the fact that, since $V^{(i)}_{-\frac{1}{2}}=V^{(i)}_{\frac{1}{2}}$, then $\chi_{\frac{1}{2}}^{(i)}(x,E)=
\chi_{-\frac{1}{2}}^{(i)}(x,E)$
and $\phi_{\frac{1}{2}}^{(i)}(x,E)=\phi_{-\frac{1}{2}}^{(i)}(x,E)$. Note that $\tilde v^{(i)}(x,E)$
belongs to an invariant subspace of the monodromy matrix corresponding to lower weights. Equation \eqref{eq:QQtilde}
is then obtained by equating the component in $W^{(i)}$
in the $\psi$-system \eqref{eq:031201} by means of equations \eqref{psigrande},
\eqref{questoeunlabelconunnomeumanomauncuorenichilista}
and the normalization \eqref{eq:chinormalization}.
\end{proof}
We now derive the $Q$-system (Bethe Ansatz) associated to the Lie algebras of $ADE$ type.
Let us denote by $E_i\in\mb C$ any zero of $Q^{(i)}(E;\ell)$.
Evaluating equation \eqref{eq:QQtilde} at $E=\Omega^{\pm\frac{1}{2}}E_i$ we get
the following relations
\begin{align}
\begin{split}\label{quasifatta}
&\prod_{j \in I}  Q^{(j)}(\Omega^{\frac{1}{2}} E;\ell)^{B_{ij}} =
- \omega^{\frac{1}{2} \theta_i } Q^{(i)}(\Omega E_i;\ell)\widetilde{Q}^{(i)}(E_i;\ell)\,,
\\
&\prod_{j \in I} Q^{(j)}(\Omega^{-\frac{1}{2}} E;\ell)^{B_{ij}}
=\omega^{-\frac{1}{2} \theta_i } Q^{(i)}(\Omega^{-1} E_i;\ell)\widetilde{Q}^{(i)}(E_i;\ell)
\,.
\end{split}
\end{align}
Assuming that for generic $\ell\in\mf h$
the functions $Q^{(i)}(E;\ell)$ and $\widetilde{Q}^{(i)}(E;\ell)$
do not have common zeros, and taking the ratio of the two identities in \eqref{quasifatta},
we get, for any zero $E_i$ of $Q^{(i)}(E;\ell)$, the $Q$-system
\begin{equation}\label{eq:Qsystemthm}
\prod_{j = 1}^{n} \Omega^{\beta_jC_{ij}} \frac{Q^{(j)}\Big(\Omega^{\frac{C_{ij}}{2}}E^*\Big)}{Q^{(j)}
\Big(\Omega^{-\frac{C_{ij}}{2}}E^*\Big)}=-1
\,,
\end{equation}
where $C=\left(C_{ij}\right)_{i,j\in I}$ is the Cartan matrix of the Lie algebra $\mf g$,
 and
$$
\beta_j=\frac{1}{2 M h^\vee}\sum_{i\in I}(C^{-1})_{ij}\theta_i 
=\frac{1}{2 M h^\vee}w_\ell(\omega_j)(\ell+h)\,,
\qquad
j\in I
\,.
$$
\begin{remark}
The construction of the $Q\widetilde{Q}$ system \eqref{eq:QQtilde}
works for any affine Kac-Moody algebra.
However,  formula \eqref{eq:Qsystemthm}
holds only in the ADE case.
In order to define the analogue of equation \eqref{eq:Qsystemthm} for the
Cartan matrix of a  non-simply laced algebra $\mf{g}$,
one has to consider, as suggested in \cite{FF11},
a connection with values in the Langlands dual of the affine Lie algebra $\widehat{\mf{g}}$.
\end{remark}

\subsection{The case \texorpdfstring{$\ell=0$}{l=0}}\label{sec:l=0}
The case of $\ell=0$ is not generic and thus the $Q$ functions cannot be defined as in Section \ref{sec:Q}.
It is however straightforward to generalize the generic case to $\ell=0$, either directly or as the limit $\ell \to 0$.
Indeed, we can take the limit $\ell \to 0$ along sequences such that the element of the Weyl group $w_\ell=w$ is fixed. In this way,
the limit solutions $\chi^{(i)}(x,E,\ell=0),\phi^{(i)}(x,E,\ell=0)$ of \eqref{20141125:eq1} satisfy the Cauchy problem
\begin{align}
 \chi^{(i)}_k(x,E,\ell=0)= v^w_i \, , \qquad
 \phi^{(i)}_k(x,E,\ell=0)=f^w_i v^w_i \label{eq:chil0} \, .
\end{align}
Here $v_i^w$ and $f^w_i$ are respectively the highest weight vector of $V^{(i)}$ and the negative Chevalley generator corresponding to the
element $w$ of the Weyl group. Equations \eqref{eq:chil0} define uniquely the two solutions.
Then the spectral determinants $Q^{(i)}$, $\tilde{Q}^{(i)}$ are exactly the coefficients of $\Psi^{(i)}(x=0,E)$ with respect to
the basis element $v^w_i,f^w_i v_i$. Note that this implies 
a certain freedom in the choice of the spectral determinants $Q^{(i)}$, $\tilde{Q}^{(i)}$. In fact, the freedom in the choice of $Q^{(i)}$
corresponds to the elements of the orbit of $\omega_i$ under the Weyl group action, while the freedom in the choice of $\tilde{Q}^{(i)}$
corresponds to the elements of the orbit of the ordered pair $(\omega_i,\alpha_i)$ under the same group action.

For any choice of $w_\ell=w$, the functions $Q^{(i)}$, $\tilde{{Q}}^{(i)}$ satisfy
\eqref{quasifatta} with $\ell=0,w_\ell=w$ and the corresponding Bethe Ansatz equation \eqref{eq:Qsystemthm}, provided they do not have common
zeros.

\section{\texorpdfstring{$\mf{g}$}{g}-Airy function}\label{app:airy}
We consider in more detail the special case of \eqref{20141125:eq1} with a linear potential $p(x,E)=x$ and with $\ell=0$,
which is particularly interesting because the subdominant solutions and the spectral determinants have an integral representation.  
Since in the $\mf{sl}_2$ case the subdominant solution coincides with the Airy function, we call the 
solution obtained in the general case the $\mf g$-Airy function. The classical Airy function models the local behavior
of the subdominant solution of the Schr\"odinger equation with a generic potential close to a turning point 
\cite{fedoryuk93}; we expect the $\mf g$-Airy function to play a similar role for equation 
\eqref{20141125:eq1}.

\subsection{\texorpdfstring{$Q$}{Q} functions}
In the special case of \eqref{20141125:eq1} with a linear potential $p(x,E)=x$ and with $\ell=0$, it is straightforward to compute the
$Q$ functions. By the discussion in Section \ref{sec:l=0}, in order to define the $Q$ functions we need to chose
an element $w$ of the Weyl group and thus the distinguished element $v^w_i$ of the basis of $V^{(i)}$.
The spectral $Q^{(i)}(E)$ is then defined as the coefficient with respect to
$v_i^w$ of $\Psi^{(i)}(x=0,E)$.
Due to the special form of the potential we have
\begin{equation*}
 \Psi^{(i)}(x=0,E)=\Psi^{(i)}(-E,0) \, .
\end{equation*}
and thus, using formula \eqref{eq:21ott04}, we get
\begin{equation}\label{eq:asymptoticQ}
 Q^{(i)}(E) = e^{-\lambda^{(i)}\frac{h^\vee}{h^\vee +1}|E|^{\frac{h^\vee+1}{h^\vee}}} C_i \big( 1+o(1) \big) \mbox{ as } E \to -\infty
 \,,
\end{equation}
where $C_i$ is the coefficient of the vector $|E|^{-h}\psi^{(i)}$ with respect to the basis vector $v^w_i$. 
It is important to note that the asymptotic behavior of the functions $Q^{(i)}$ is expressed via the eigenvalues $\Lambda^{(i)}$
which coincide with the masses of the (classical) affine Toda field theory, as we proved in Proposition \ref{prop:pisa2}.
This is precisely the same behavior predicted for the vacuum eigenvalue $Q$ of the corresponding Conformal Field
Theory \cite{reshetikhin87}.

We conclude that the ODE/IM correspondence depends on the relation among the element $\Lambda \in \widehat{\mf g}$,
the Perron-Frobenius eigenvalue of the incidence matrix $B$ and the masses of the affine Toda field theory.

\begin{remark}
Note that formula \eqref{eq:asymptoticQ} does not coincide with the 
asymptotic behavior conjectured in \cite[Eq 2.6]{dorey07} for a general potential.
In fact, J. Suzuki communicated us \cite{junjiprivate} that the latter formula is conjectured to hold
only for potentials such that $\frac{M+1}{M h^\vee} <1$, in which case the Hadamard's product \cite[Eq 2.7]{dorey07}
converges. Clearly in case of a linear potential $\frac{M+1}{M h^\vee} =\frac{h^\vee+1}{h^\vee}>1$.
\end{remark}
\begin{remark}
A connection between generalized Airy equations and integrable systems appears, among other places,  in  \cite{kacschwarz91}, where the the orbit of the KdV topological tau--function in the Sato Grassmannian is described. We note that an operator of the form \eqref{20141020:eq1} -- specialized to the case of algebras of type $\mf{sl}_n$ and for $\ell= - h / h^\vee$, $M=1 / h^\vee$ --  also appears in that paper. It would be interesting to obtain a more precise relation between these results and those of the present paper.
\end{remark}

\subsection{Integral Representation}
Given an evaluation representation of $\widehat {\mf g}$, we look for
solutions of the equation
\begin{equation}\label{eq:31ott01}
\Psi'(x) + \big( e +x e_0 \big) \Psi(x)=0
\,,
\end{equation}
in the form
\begin{equation}\label{eq:31ott02}
\Psi(x)=\int_{ \mathit{c}} e^{-xs} \Phi(s) ds\,,
\end{equation}
where $\Phi(s)$ is an analytic function and $\mathit{c}$ is some path in the complex plane. 
Differentiating and integrating by parts we get
\begin{equation}\label{oltrelultima}
\int_{\mathit{c}}  e^{-xs} \big( -s + e + e_0 \frac{d}{ds}\big) \Phi(s) ds
+ e^{-xs }e_0 \Phi(s) \big|_{\mathit{c}}=0
\,, 
\end{equation}
where the last term is the evaluation of the integrand at the end points of the path.
We are therefore led to the study of the simpler linear equation
\begin{equation}\label{eq:3nov01}
\big(- s + e + e_0 \frac{d}{ds}\big) \Phi(s)=0
\,,
\end{equation}
which we solve below for two important examples.

\subsection{\texorpdfstring{$A_{n-1}$}{An} in the representation \texorpdfstring{$V^{(1)}$}{V^{(1)}}}
In this case equation \eqref{eq:31ott01}  is already known in the literature as the $n$-Airy equation \cite{fedoryuk93}.
Using the explicit form of Chevalley generators $e_i$, $i=0,\dots,n-1$, which is provided in Appendix \ref{app:An}, and denoting by $\Phi_i(s)$ the component of
$\Phi(s)$ in the standard basis of $\mb{C}^n$,
we find the general solution of \eqref{eq:3nov01} as
\begin{equation*}
\Phi_1(s)=\kappa\,e^{\frac{s^{n+1}}{n+1}}\,,
\qquad\Phi_i(s)=s^{i-1}\Phi_1(s)\,,
\quad i=2,\dots,n
\,,
\end{equation*}
where $\kappa$ is an arbitrary complex number.
If we choose the path $\mathit{c}$ as the one connecting
$e^{-\frac{\pi i}{n+1}} \infty$ with $ e^{\frac{\pi i}{n+1}} \infty$, then the integral formula \eqref{eq:31ott02}
is well-behaved, and the boundary terms in \eqref{oltrelultima} vanish.
Thus, the $\mf{sl}_n$-Airy function has the integral representation
\begin{equation}\label{eq:3nov02}
\Psi_j(x)
=\frac{1}{2 \pi i} \int_{e^{-\frac{\pi i}{n+1}} \infty}^{e^{+\frac{\pi i}{n+1}} \infty}
s^{j-1} e^{-xs+\frac{s^{n+1}}{n+1}} ds\,, \qquad j=1,\dots, n
\,.
\end{equation}
In case $n=2$, the latter definition coincides with the definition of the standard Airy function.
By the method of steepest descent we get the following asymptotic expansion of $\Psi$ for $x \gg 0$:
\begin{equation}\label{eq:3nov03}
\Psi_j(x) \sim \sqrt{\frac{1}{2 \pi n}}
x^{\frac{2j-1-n}{2n}}   e^{-\frac{n}{n+1}x^{\frac{n+1}{n}}}\,, \qquad j=1,\dots, n
\,.
\end{equation}
Due to Theorem \ref{thm:asymptotic}, the $\mf{sl}_n$-Airy function coincides with the fundamental solution $\Psi^{(1)}$ of the
linear ODE \eqref{eq:31ott01}. In fact the asymptotic behavior of the Airy function coincides with the asymptotic behavior
of the function $\Psi^{(1)}$ (for this computation one needs to use the explicit formula for $h$ which is given in equation (\ref{eq:hAn}) ).
The solutions $\Psi^{(1)}_k$,  $k \in \mb{Z}$, are obtained by integrating \eqref{eq:31ott02} along the contour
obtained rotating $\mathit{c}$ by $e^{\frac{2 k \pi i }{n+1}}$.

\subsection{\texorpdfstring{$D_n$}{Dn} in the standard representation}
The second  example is the $D_n$-Airy function in the representation $V^{(1)}$. 
Using the Chevalley generators $e_i$, $i=0,\dots,n$, as described in Appendix \ref{app:Dn}, and denoting 
 $\Phi_j$ the $j$th-component of $\Phi$ in the standard basis of $\mb{C}^{2n}$, we find that the
 general solution of \eqref{eq:3nov01} reads
\begin{align*} 
&  \Phi_1=\kappa\,s^{-\frac{1}{2}} e^{\frac{s^{2n-1}}{2n-1}} \label{eq:3nov06} \\ \nonumber
& \Phi_j=s^{j-1} \Phi_1(s)\,,\qquad j\leq n-1\,, \\ \nonumber
& \Phi_{n}=\frac{s^{n-1}}{2}\Phi_1(s)\,,\quad \Phi_{n+1}=s^{n-1}\Phi_1(s)\,, \\ \nonumber
& \Phi_{n+j}=s^{n+j-2} \Phi_1(s) \,,\qquad j \leq n-1\,,  \\ \nonumber
& \Phi_{2n}=\left( \frac{s^{2n-2}}{2} + \frac{1}{4 s} \right) \Phi_1(s) \,.
\end{align*}
where $\kappa \in\mb C$ is an arbitrary constant. Let $\mathit{c}$ be the path connecting
$e^{-\frac{\pi i}{2n-1}} \infty$ with $ e^{\frac{\pi i}{2n-1}} \infty$. Then along $\mathit{c}$
the integral formula \eqref{eq:31ott02} is well-behaved and it thus defines a solution of \eqref{eq:31ott01}
that we call $D_n$-Airy function. Moreover, the boundary terms in \eqref{oltrelultima} vanish. 
As in the case $A_n$, the standard method of steepest descent shows that such solution is exactly the fundamental solution
$\Psi^{(1)}$ to equation
\eqref{eq:31ott01}.
The solutions $\Psi^{(1)}_k , \, k \in \bb{Z}$, are obtained by integrating \eqref{eq:31ott02} along the contour obtained rotating
$\mathit{c}$ by $e^{\frac{k \pi i }{2n-1}}$.


{\appendix\section{Action of \texorpdfstring{$\Lambda$}{Lambda}
in the fundamental representations}\label{app:PF}
In this appendix we give an explicit description of the maximal eigenvalues of $\Lambda$ in the fundamental
representations and of the corresponding eigenvector.

\subsection{The $A_n$ case}\label{app:An}
The simple Lie algebra of type $A_n$, $n\geq 1$, can be realized as the algebra of $(n+1)\times(n+1)$
traceless matrices
$$
\mf g=\mf{sl}_{n+1}=\{A\in\Mat_{n+1}(\mb C)\mid \tr A=0\}
\,,
$$
where the Lie bracket is the usual commutator of matrices.
The dual Coxeter number of $\mf g$ is $h^\vee=n+1$.
Let us consider the following Chevalley generators of $\mf g$ ($i\in I=\{1,\dots,n\}$):
$$
f_i=E_{i+1,i}\,,
\qquad
h_i=E_{ii}-E_{i+1,i+1}\,,
\qquad
e_i=E_{i,i+1}
\,,
$$
where $E_{ij}$ denotes the elementary matrix with $1$ in position $(i,j)$ and $0$ elsewhere.
It is well-known that the representation $L(\omega_1)$ is given by the natural
action of $\mf g$ on $L(\omega_1)=\mb C^{n+1}$.
Moreover, we have that
$$
L(\omega_i)=\bigwedge^i L(\omega_1)\,,
\qquad
i\in I\,.
$$
We denote by $u_j$, $j=1,\dots,n+1$, the standard basis
of $\mb C^{n+1}$, and by $v_i$, $i\in I$, the highest weight vector of the representation
$L(\omega_i)$. Then, we have:
$$
v_i=u_1\wedge u_2\wedge \dots\wedge u_i
\,.
$$
The set of Chevalley generators for $\widehat{\mf g}$ is obtained by adding to the Chevalley generators of
$\mf g$ the following elements:
$$
f_0=E_{1,n+1}t^{-1}\,,
\qquad
h_0=2c-E_{11}+E_{n+1,n+1}\,,
\qquad
e_{0}=E_{n+1,1}t\,.
$$
The element $h \in \mf h$ satisfying relations (\ref{20141020:eq2}) is
\begin{equation}\label{eq:hAn}
 h=\diag \big( -\frac{n}{2},-\frac{n-1}{2},\dots, \frac{n-1}{2},\frac{n}{2}\big) \, .
\end{equation}

Recall that $\Lambda=e_0+e_1+\dots+e_n$ and,
from Example \ref{exa:An}, that
$V^{(i)}=\bigwedge^i L(\omega_1)_{\frac{i-1}{2}}^{(1)}$, for $i\in I$.
In particular, $V^{(1)}=\mb C^{n+1}$,  and we set
\begin{equation}\label{app:psi_1}
\psi^{(1)}
=\sum_{j=1}^{n+1}u_j
\in V^{(1)}
\,.
\end{equation}
Then, it is easy to check that $\Lambda \psi^{(1)}=\psi^{(1)}$.
By Proposition \ref{prop:pisa2}(a), $\psi^{(1)}\in V^{(1)}$ is the unique (up to a constant) eigenvector
corresponding to the maximal eigenvalue $\lambda^{(1)}=1$. 
Furthermore, using equation \eqref{eq:Blambda}, we can check that for every $i\in I$,
\begin{equation}\label{eq:lambda_A}
\lambda^{(i)}
=\frac{\sin\left(\frac{i\pi}{n+1}\right)}{\sin\left(\frac{\pi}{n+1}\right)}\,,
\end{equation}
is a maximal eigenvalue of $\Lambda$ in the representation $V^{(i)}$. Using equations \eqref{20150108:eq6}
and \eqref{20150108:eq3}, the corresponding
eigenvector is easily checked to be
$$
\psi^{(i)}
=\psi_{-\frac{i-1}2}^{(1)}\wedge\psi_{-\frac{i-3}2}^{(1)}
\wedge\dots\wedge\psi_{\frac{i-3}2}^{(i)}\wedge\psi_{\frac{i-1}2}^{(1)}
\in V^{(i)}
\,.
$$

\subsection{The $D_n$ case}\label{app:Dn}
Let $n\in\mb Z_+$, and consider the involution on the set
$\{1,\dots,2n\}$ defined by $i\to i^\prime=2n+1-i$.
Given a matrix $A=\left(A_{ij}\right)_{i,j=1}^{2n}\in\Mat_{2n}(\mb C)$ we define its anti-transpose
(the transpose with respect to the antidiagonal) by
$$
A^{\at}=\left(A_{ij}^{\at}\right)_{i,j=1}^{2n}\,,
\quad
\text{where}
\quad A_{ij}^{\at}=A_{j^\prime i^\prime}
\,.
$$
Let us set
$$
S=\sum_{k=1}^n(-1)^{k+1}
\left(E_{kk}+E_{k^\prime k^\prime}
\right)\,.
$$
Following \cite{DS85}, the simple Lie algebra of type $D_n$ can be realized as the algebra
$$
\mf g=\mf o_{2n}=\{A\in\Mat_{2n}(\mb C)\mid AS+SA^{\at}=0\}
\,,
$$
where the Lie bracket is the usual commutator of matrices.
The dual Coxeter number of $\mf g$ is $h^\vee=2n-2$.
For $i,j\in I=\{1,2,\dots,n\}$, we define
$$
F_{ij}=E_{ij}+(-1)^{i+j+1}E_{j^\prime i^\prime}\,,
\qquad
\widetilde{F}_{ij}=E_{ij^\prime}+(-1)^{i+j+1}E_{ji^\prime}
\,,
$$
and we consider the following Chevalley generators of $\mf g$:
\begin{align*}
&f_i=F_{i+1,i}\,,
&&h_i=F_{ii}-F_{i+1,i+1}\,,
&&e_i=F_{i,i+1}\,,
&i=1,\dots,n-1\,,
\\
&f_n=2\widetilde{F}_{n^\prime,(n-1)^\prime}\,,
&&h_n=F_{n-1,n-1}+F_{n,n}\,,
&&e_n=\frac12\widetilde{F}_{n-1,n}
\,.
\end{align*}
It is well-known that the representation $L(\omega_1)$ is given by the natural
action of $\mf g$ on $L(\omega_1)=\mb C^{2n}$.
Moreover we have that
$$
L(\omega_i)=\bigwedge^iL(\omega_1)\,,\qquad i=1,\dots,n-2\,,
$$
while $L(\omega_{n-1})$ and $L(\omega_n)$ are the so-called half-spin representations of $\mf g$.
We denote by $u_j$, $j=1,\dots,2n$, the standard basis
of $\mb C^{2n}$, and by $v_i$, $i\in I$, the highest weight vector of the representation
$L(\omega_i)$. It is well-known that
$$
v_i=u_1\wedge u_2\wedge \dots\wedge u_i
\,,
\qquad i=1,\dots,n-2
\,.
$$

The set of Chevalley generators for $\widehat{\mf g}$ is obtained by adding to the Chevalley generators of
$\mf g$ the following elements:
$$
f_0=2\widetilde{F}_{(2n)^\prime,(2n-1)^\prime}t^{-1}\,,
\qquad
h_0=2c-F_{11}-F_{22}\,,
\qquad
e_{0}=\frac12\widetilde{F}_{2n-1,2n}t\,.
$$
Recall that $\Lambda=e_0+e_1+\dots+e_n$ and,
from Example \ref{exa:An}, that
$V^{(i)}=\bigwedge^i L(\omega_1)_{\frac{i-1}{2}}^{(1)}$, for $i=1,\dots,n-2$,
$V^{(n-1)}=L(\omega_{n-1})_{\frac n2}$ and
$V^{(n)}=L(\omega_{n})_{\frac n2}$ are the evaluation representations at $\zeta=(-1)^n$
of the half-spin representations.
In particular, $V^{(1)}=\mb C^{2n}$,  and we set
$$
\psi^{(1)}=\sum_{j=1}^{n-1}\left(u_j+u_{n+j}\right)
+\frac12\left(u_n+u_{2n}\right)
\,.
$$
Then, it is easy to check that $\Lambda \psi^{(1)}=\psi^{(1)}$.
By Proposition \ref{prop:pisa2}(a), $\psi^{(1)}\in V^{(1)}$ is the unique (up to a constant) eigenvector
corresponding to the maximal eigenvalue $\lambda^{(1)}=1$. 
Furthermore, using equation \eqref{eq:Blambda}, we can check that for every $i=1,\dots,n-2$,
$$
\lambda^{(i)}
=\frac{\sin\left(\frac{i\pi}{2n-2}\right)}{\sin\left(\frac{\pi}{2n-2}\right)}\,,
$$
is the maximal eigenvalue of $\Lambda$ in the representation $V^{(i)}$, and the corresponding
eigenvector is easily checked to be
$$
\psi^{(i)}
=\psi_{-\frac{i-1}2}^{(1)}\wedge\psi_{-\frac{i-3}2}^{(1)}
\wedge\dots\wedge\psi_{\frac{i-3}2}^{(i)}\wedge\psi_{\frac{i-1}2}^{(1)}
\in V^{(i)}
\,,
$$
as follows from equations \eqref{20150108:eq6} and \eqref{20150108:eq3}.
For the fundamental representations $V^{(n-1)}$ and $V^{(n)}$
(see \cite{FH91} for the definition), it is easy to check, using again equation \eqref{eq:Blambda},
that
$$
\lambda^{(n-1)}=\lambda^{(n)}=\frac{1}{2\sin\left(\frac{\pi}{2n-2}\right)}
\,.
$$
Finally, we note that there is another important irreducible representation for the Lie algebra
$\mf g$, which we denote
by $U^{(n-1)}=\bigwedge^{n-1}V^{(1)}_\frac n2$.
It is not hard to show that $\widetilde\lambda^{(n-1)}=2\lambda^{(n-1)}$
is the maximal eigenvalue for the action of $\Lambda$ on $U^{(n-1)}$.
Its corresponding eigenvector is
\begin{equation*}
 \widetilde{\psi}^{(n-1)}=\psi_{-\frac{n-2}2}^{(1)}\wedge\psi_{-\frac{n-3}2}^{(1)}
\wedge\dots\wedge\psi_{\frac{n-3}2}^{(i)}\wedge\psi_{\frac{n-2}2}^{(1)}
\,.
\end{equation*}
Moreover, the following isomorphism of representation (see \cite{FH91})
\begin{equation*}
V^{(n-1)}\otimes V^{(n)}\cong U^{(n-1)}\oplus V^{(n-3)}\oplus\dots
\,,
\end{equation*}
implies that there exists a morphism of representations $\widetilde{m}$ such that
\begin{equation}\label{eq:ultimogiorno}
\widetilde{m}\left(  \widetilde{\psi}^{(n-1)} \right)=\psi^{(n-1)}\otimes \psi^{(n)} \, .
\end{equation}

\subsection{The \texorpdfstring{$E$}{E} case}\label{app:E}
In Table \ref{tabella} we give the explicit
form of the eigenvalues $\lambda^{(i)}$'s, obtained from \eqref{eq:Blambda}. The corresponding eigenvectors
are computed in a GAP \cite{GAP4} file available upon request to the authors. 

\begin{table}[H]
\normalsize
\caption{Maximal eigenvalues for simple Lie algebras of $E$ type} \label{tabella} 
\begin{center}
\begin{tabular}{c||c|c|c} 
& $E_6$ & $E_7$ & $E_8$
\\
\hline\hline
\phantom{$\Bigg($}
$\lambda^{(1)}$ & 1  & 1&  1
\\
\hline
\phantom{$\Bigg($}$\lambda^{(2)}$ &
$\frac{\sin\left(\frac{\pi}{6}\right)}{\sin\left(\frac{\pi}{12}\right)}$ &
$\frac{\sin\left(\frac{\pi}{9}\right)}{\sin\left(\frac{\pi}{18}\right)}$ &
$\frac{\sin\left(\frac{2\pi}{30}\right)}{\sin\left(\frac{\pi}{30}\right)}$
\\
\hline
\phantom{$\Bigg($}$\lambda^{(3)}$ & $\frac{\sin\left(\frac{\pi}{4}\right)}{\sin\left(\frac{\pi}{12}\right)}$&
$\frac{\sin\left(\frac{\pi}{6}\right)}{\sin\left(\frac{\pi}{18}\right)}$&
$\frac{\sin\left(\frac{\pi}{10}\right)}{\sin\left(\frac{\pi}{30}\right)}$
\\
\hline
\phantom{$\Bigg($}
$\lambda^{(4)}$ & $\frac{\sin\left(\frac{\pi}{4}\right)}{\sin\left(\frac{\pi}{12}\right)}$&
$\frac{\sin\left(\frac{2\pi}{9}\right)}{\sin\left(\frac{\pi}{18}\right)}$&
$\frac{\sin\left(\frac{2\pi}{15}\right)}{\sin\left(\frac{\pi}{30}\right)}$
\\
\hline
\phantom{$\Bigg($}$\lambda^{(5)}$ &$\frac{\sin\left(\frac{\pi}{6}\right)}{\sin\left(\frac{\pi}{12}\right)}$ &
$\frac{\sin\left(\frac{2\pi}{9}\right)}{\sin\left(\frac{\pi}{9}\right)}$
 &
$\frac{\sin\left(\frac{\pi}{6}\right)}{\sin\left(\frac{\pi}{30}\right)}$
\\
\hline\phantom{$\Bigg($}
$\lambda^{(6)}$ & $1$  &
$\frac{\sin\left(\frac{\pi}{9}\right)\sin\left(\frac{2\pi}{9}\right)}{\sin\left(\frac{\pi}{6}\right)\sin\left(\frac{\pi}{18}\right)}$ &
$\frac{\sin\left(\frac{\pi}{6}\right)}{\sin\left(\frac{\pi}{15}\right)}$
\\
\hline\phantom{$\Bigg($}
$\lambda^{(7)}$ & & $\frac{\sin\left(\frac{2\pi}{9}\right)}{\sin\left(\frac{\pi}{6}\right)}$ &
$\frac{\sin\left(\frac{\pi}{6}\right)\sin\left(\frac{\pi}{15}\right)}{\sin\left(\frac{\pi}{10}\right)\sin\left(\frac{\pi}{30}\right)}$
\\
\hline
\phantom{$\Bigg($}
$\lambda^{(8)}$ & & & $\frac{\sin\left(\frac{\pi}{6}\right)}{\sin\left(\frac{\pi}{10}\right)}$
\end{tabular}
\end{center}
\end{table}


\def\cprime{$'$} \def\cprime{$'$} \def\cprime{$'$} \def\cprime{$'$}
  \def\cprime{$'$} \def\cprime{$'$} \def\cprime{$'$} \def\cprime{$'$}
  \def\cprime{$'$} \def\cprime{$'$} \def\cydot{\leavevmode\raise.4ex\hbox{.}}
  \def\cprime{$'$} \def\cprime{$'$} \def\cprime{$'$}

\end{document}